\documentclass[a4paper,fleqn]{cas-sc}

\usepackage[numbers]{natbib}
\usepackage{amsthm}
\usepackage{pgfplots}
\pgfplotsset{compat=1.7}
\usepackage{subfig}
\usepackage{mleftright} 
\mleftright % redefine \left as \mleft and \right as \mright.
\usepackage[linesnumbered,ruled,lined]{algorithm2e} %algpseudocode
\DontPrintSemicolon

\def\tsc#1{\csdef{#1}{\textsc{\lowercase{#1}}\xspace}}
\tsc{WGM}
\tsc{QE}

\SetKwComment{Comment}{$\triangleright$\ }{}

\newcommand{\bigO}{\mathcal{O}}
\newcommand{\fun}[1]{\textsc{#1}}
\newcommand{\longprogressions}{\fun{LongestArithProg*}}
\newcommand{\isIMAP}{\fun{IsIMAP}}
\newcommand{\isIMAPprime}{\fun{IsIMAP*}}
\newcommand{\extend}{\fun{Extend}}

\newcommand{\factors}{\fun{Factors}}
\newcommand{\factorsprime}{\fun{Factors*}}

\newcommand{\apt}{\ensuremath{[t_0,d,k]}}
\newcommand{\Sap}{\ensuremath{S_{\text{ap}} }}
\newcommand{\Sapm}{\ensuremath{S_{\text{ap}}^{\text{max}} }}

\newcommand{\seq}{\ensuremath{(s_i)}}

\newcommand{\nlast}{\refstepcounter{AlgoLine}\nlset{\theAlgoLine\rlap{*}}}
\let\oldnl\nl
\newcommand{\nlnonumber}{\renewcommand{\nl}{\let\nl\oldnl}}
\newtheorem{theorem}{Theorem}
\newtheorem{lemma}[theorem]{Lemma}
\newdefinition{rmk}{Remark}
\newproof{pf}{Proof}
\newcommand{\proofofref}{}
\newproof{zproofof}{Proof of \proofofref}
\newenvironment{proofof}[1]
 {\renewcommand{\proofofref}{#1}\zproofof}
 {\endzproofof}

\newtheorem{corollary}{Corollary}[theorem]
\newtheorem{proposition}{Proposition}[theorem]
\begin{document}
\let\WriteBookmarks\relax
\def\floatpagepagefraction{1}
\def\textpagefraction{.001}

\shorttitle{Enumerating IMAPs}    

\shortauthors{Bemman et al.}  

\title [mode = title]{Enumerating Inclusion-Maximal Arithmetic Progressions}  




\author[1]{Brian Bemman}%[<options>]

\cormark[1]

\fnmark[1]

\ead{brian.m.bemman@durham.ac.uk}



\affiliation[1]{organization={Department of Computer Science, Durham University},
            addressline={Stockton Rd}, 
            city={Durham},
%          citysep={}, % Uncomment if no comma needed between city and postcode
            postcode={DH1 3LE}, 
            %state={},
            country={United Kingdom}}

\author[1]{Maximilien Gadouleau}%[]

% Footnote of the second author

% Email id of the second author
\ead{m.r.gadouleau@durham.ac.uk}



\author[2]{Oliver W. Gnilke}%[]

% Footnote of the third author
%\fnmark[3]

% Email id of the third author
\ead{owg@math.aau.dk}



\affiliation[2]{organization={Department of Mathematical Sciences, Aalborg University},
            addressline={Thomas Manns Vej 23}, 
            city={Aalborg Øst},
            citysep={}, % Uncomment if no comma needed between city and postcode
            postcode={9220},
            %state={},
            country={Denmark}}

\author[1]{George B. Mertzios}%[]

% Footnote of the fourth author
%\fnmark[4]

% Email id of the fourth author
\ead{george.mertzios@durham.ac.uk}



\cortext[1]{Corresponding author}



\begin{abstract}
We present a simple $\bigO\left( n^2 \frac{ \log N }{ \log \log N } + N \right)$ enumeration algorithm for solving a problem from mathematical and computational music analysis where, given a strictly increasing integer sequence, $S$, with $n$ entries and maximum value $N$, the task is to enumerate all $m$ \textit{inclusion-maximal arithmetic progressions (IMAPs)} in this sequence. An IMAP is a subsequence, $S' \subseteq S$ with $k>2$ integers, in which (i) the difference between any two consecutive integers is the same number, $d$ (i.e., $S'$ is an \textit{arithmetic progression}), (ii) $S'$ cannot be further extended to the left or to the right with any additional integers from $S$ while still remaining an arithmetic progression (i.e., $S'$ is a \textit{maximal} arithmetic progression), and (iii) there is no other maximal arithmetic progression, $S'' \subseteq S$, which \textit{properly} contains $S'$ (i.e., $S'$ is an \textit{inclusion-maximal} arithmetic progression). We further provide proofs for the expected number of IMAPs in random integer sequences, $S$, and a bound on their order of growth. Finally, we provide empirical experiments comparing both (a)~the practical running time performance of the proposed algorithm against that of a previously known  algorithm which has higher time complexity $\bigO(N^{2+o(1)}n)$, and (b)~the actual enumerated number of IMAPs to that of their mathematically expected number. Notably, the proposed algorithm demonstrates a significant improvement in running time over the previously known algorithm, and in immediate practical applications, will allow for more efficient analysis of large and rhythmically complex musical pieces.
\end{abstract}



% Keywords

\maketitle

% Main text
%%%%%%%%%%%%%%%%%%%%%%%%%%%%%%%%%%%%%%%%%%%%%%
\section{Introduction}

Let $S = \langle s_0, s_1 \ldots, s_{n-1} \rangle$ be a strictly increasing sequence of positive integers. An \textit{inclusion-maximal arithmetic progression (IMAP)} \cite{Nestke2001} is a subsequence, $t_0, \dots, t_{k-1}$, of non-trivial length, $k>2$, such that 

\begin{enumerate}
    \item $t_{j+1} = t_j + d$ for $0 \le j < k-1$ where $d>0$ is the constant \textit{common difference} between any two consecutive integers (\textit{arithmetic progression});
    \item neither $t_0 - d$ nor $t_{k-1} + d$ are elements of $S$ (\textit{maximal} arithmetic progression); and
    \item it is not \textit{properly} contained within any other subsequence of $S$ that satisfies (1) and (2) (\textit{inclusion-maximal} arithmetic progression).
\end{enumerate}

\noindent By way of an example, consider the following sequence, $S=\langle 1,2,3,4,5,6,8\rangle$, which has exactly three IMAP subsequences, namely, $S'_1=\langle 1,2,3,4,5,6\rangle, \, S'_2=\langle 2,4,6,8\rangle$, and $S'_3=\langle 2,5,8\rangle$. Note that while the subsequence, $\langle 1,3,5\rangle$, of $S$ is a maximal arithmetic progression (MAP) as it cannot be further extended, it is not \textit{inclusion-maximal} because it is contained within another IMAP subsequence, namely, $S'_1=\langle \underline{\textbf{1}},2,\underline{\textbf{3}},4,\underline{\textbf{5}},6\rangle$. While (1) and (3) above are sufficient to define an IMAP, (2) is necessary to distinguish between two closely related problems in computing: enumerating all MAPs (i.e., (1) and (2) above), which forms the basis of our proposed algorithm, and enumerating all IMAPs, which is a special case of the former. It is important to also note that in enumerating IMAPs, we are concerned here with \textit{all} IMAPs rather than the single longest arithmetic progression.
%which form a subset of all MAPs

Computational and mathematical methods have long been applied to music (e.g., use of the AllDifferent global constraint from constraint programming to solve the all-interval series problem in music \cite{Anders2011}), while equally, problems originating from music have proved novel within computational and mathematical domains (e.g., special cases of set-covering from music analysis and composition \cite{Tanaka2016a,Tanaka2016b,Bemman2019}). IMAPs are a clear example of both that were introduced in \cite{Nestke2001} through simplicial complexes and for computing a so-called \textit{Inner Metric Analysis}. IMAPs are an important feature of Inner Metric Analysis used to carry out various computationally- and mathematically-based music analysis tasks largely pertaining to metrical structure, such as meter induction \cite{BasDeHaas2016}, detecting and predicting rhythmic syncopation \cite{Volk2008,Koops2015,Bemman2024}, and predicting musical genre \cite{Chew2005}. The relevance of IMAPs to music is relatively straightforward insofar as being able to investigate all equally-spaced points in time within a sequence of $n$ discrete musical events can provide insight into the metrical regularity and structure directly from the music itself absent of any additional information that may otherwise be provided by the printed score (e.g., time signature). While a larger $N$ relative to $n$ generally indicates a more rhythmically-complex piece of music (i.e., one containing both small and more diverse rhythmic durations requiring a larger range of $N$ integers by which to properly encode them), the regularity and non-random nature of music presents an interesting case since the density of sequences remains high yet even when it decreases, most events are near-evenly spaced such that relatively high numbers of IMAPs still emerge.

In more directly related work, algorithms exist for finding the longest arithmetic progression and all MAPs \cite{Erickson1999} as well as proofs concerning various bounds on, and the expected number of, similar mathematical and combinatorial objects, such as $k$-term arithmetic progressions and MAPs \cite{Green2008,Kelley2023,Bloom2023a,Bloom2023b,Erickson1999}. In the case of IMAPs, however, no such proofs exist. Moreover, only a single procedural algorithm exists in open-sourced form capable of enumerating IMAPs \cite{Kranenburg2013} but is relatively inefficient, which has limited its use in large-scale music analytical tasks, and details concerning its design have so far remained unpublished outside of music. 

%%%%%%%%%%%
\subsection{Contribution}\label{sec:contribution}
In this paper, we provide the first-known efficient algorithm for enumerating IMAPs as well as prove for the first time the expected number of such progressions in random integer sequences and a bound on their order of growth. We show that our proposed algorithm is near-optimal when $N=\Theta(n)$. Additionally, we offer empirical results in the form of comparisons between the practical running time performances of the proposed algorithm and the aforementioned existing algorithm as well as between the actual enumerated number of IMAPs and their expected number as provided in our proof and exact construction. In applied terms, our contribution provides mathematical and computational music analysts with the ability to more efficiently analyze large and rhythmically complex musical pieces.

%%%%%%%%%%%
\section{Preliminaries}\label{sec:notation}

We begin by defining $S = \langle s_0, s_1, \ldots, s_{n-1} \rangle$ to be a strictly increasing sequence of $n$ positive integers drawn from the range, $\{1, 2, \ldots, N\}$ with $1 \leq s_0 < s_1 < \cdots < s_{n-1} \leq N$ where $n \leq N$. In various contexts throughout, $S$ differs only in how $N$ is fixed. In algorithmic contexts (sections~\ref{sec:previous-IMAP-algorithms} and \ref{sec:proposed-algorithm}), $S$ is the input such that $N := s_{n-1} = \max(S)$, so the upper bound of the range is tight. In probabilistic contexts (sections~\ref{sec:expected-number} and \ref{sec:experiments}), $S$ is drawn uniformly at random from all $\binom{N}{n}$ strictly increasing sequences in $\{1, 2, \ldots, N\}$, so $s_{n-1} \leq N$ in general. An IMAP can be uniquely defined by a pair of variables, $(s_i, d)$, referring to its starting element, $s_i\in S$, and its common difference, $d>0$. As done in previous work, we further denote by $l$ the number of elements minus 1 of an IMAP to fully specify it. Additionally, we will use $e$ to denote the ending element of an IMAP. The total number of IMAPs in $S$ we denote by $m$. For all algorithms, we use zero-based indexing and a left arrow, $\leftarrow$, to denote variable assignment. Ordered collections are indexable (e.g., arrays or lists) and denoted using angled brackets, $\langle x\rangle$, while unordered collections (e.g., dictionary or hashmap) are denoted using curly braces, $\lbrace x\rbrace$. Square brackets, $[x]$, are used to denote access to an arbitrary collection. For example, if $S$ is an ordered collection and $F$ is an unordered collection, then $S[d]$ denotes the $(d+1)$th element of $S$, while $F[d]$ denotes the value in $F$ for the dictionary key, $d$. Concatenation, $\oplus$, is used to denote the appending of an element (or collection) to an ordered collection. For convenience, a single period, `$.$', is used to denote access to the field of a composite-type variable (e.g., an object or struct). For example, `$p_0.d$' refers to the value of the field named `$d$' belonging to `$p_0$'. In all complexity analyses of the algorithms provided here, we will assume a standard word-RAM model, in which elementary arithmetic operations are unit-cost, rather than a logarithmic-cost model, unless where otherwise stated. Use of unordered collections and their associated access notation is merely to simplify the presentation of the pseudo-code of the proposed algorithm. In its actual implementation, we use indexable ordered collections---without altering the algorithm's complexity---to remain in line with the standard assumptions of a word-RAM model along with use of memory-efficient compressed sparse row ordered collections, flat indexing, and hoisted variables, wherever possible.

%%%%%%%%%%%%%%%%%%%%%%%%%%%%%%%%%%%%%%%%%%%%%%
\section{Related Work}\label{sec:related}

In additive number theory and combinatorics, the study of $k$-term arithmetic progressions and the various properties defining sequences of integers that either contain them or not has been of significant interest to researchers over many years \citep{Green2026}. Considerable work has been done to prove bounds on the number of $k$-term arithmetic progressions in arbitrary sequences \cite{Green2008,Kelley2023,Bloom2023a,Bloom2023b}. Green and Sisask have shown, for example, that the maximum number of 3-term arithmetic progressions in $n$-element integer sets is $\frac{n^2}{2}$ \cite{Green2008}, while more recently strong bounds have been proven for the same \cite{Kelley2023,Bloom2023a,Bloom2023b}. The more general case of proving bounds on the number of $k$-term arithmetic progressions remains challenging \cite{Bloom2023a,Bloom2023b,Green2026}. Similar work concerning bounds on MAPs is comparatively more limited, however, an upper-bound of $\bigO(\frac{n^2}{k^2})$ exists on the number of MAPs of at least length, $k$ \cite{Erickson1999}. 

In computing, the problem of finding the longest arithmetic progression in a sequence of length, $n$, has been solved using an output-sensitive dynamic programming algorithm having a time complexity of $\bigO(n^2)$, which is optimal in the 3-linear decision tree model of computation \cite{Erickson1999}. Incidentally, this algorithm can be adapted at little additional cost to enumerate all 3-term arithmetic progressions and, as has been suggested, can also be used to find all MAPs of at least length, $k$ \cite{Erickson1999}. In this same work, a similarly output-sensitive algorithm adopting a divide-and-conquer approach making use of the dynamic programming algorithm was provided for finding all such MAPs. This divide-and-conquer approach has a time complexity of $\bigO(\frac{n^2}{k} \, \log\frac{n}{k} \, \log \, k)$, which is an improvement over using the dynamic programming approach alone only for values of $k > \log \, n \, \log \, \log \, n$. In comparison to the related problems described above, considerably less work exists concerning IMAPs, and all such published work currently resides within their applications to mathematical and computational music analysis and theory \cite{Nestke2001,Chew2005,Volk2008,Koops2015,Bemman2024}.

%%%%%%%%%%%
\subsection{An Existing Algorithm for Enumerating Inclusion-Maximal Arithmetic Progressions}\label{sec:previous-IMAP-algorithms}

While there exist a few algorithms capable of enumerating IMAPs, these have been developed for applications to computational music analysis and most remain unpublished \cite{Volk2008,BasDeHaas2016}. The first known of these algorithms was originally implemented by Chris Dyer (2011) in Java and never published but has been used extensively in music-related research \cite{Volk2008,Koops2015}. A purely functional implementation, written in Haskell, has also been developed but similarly remains unpublished \cite{BasDeHaas2016}. Finally, a ported implementation of the original algorithm to C++ has since been made open source \cite{Kranenburg2013}, but details of its design have thus far never been published as part of any research.

%%%%%%%%%%%
\begin{algorithm}[!t]
\DontPrintSemicolon
\SetAlgoLined
\SetProcNameSty{textsc}
\SetProcArgSty{textsc}
\renewcommand{\SetProgSty}[1]{\renewcommand{\ProgSty}[1]{\textnormal{\csname#1\endcsname{##1}}\unskip}}%
\SetProgSty{textsc}
\SetKwFunction{Fcontained}{\contained}
\SetKwData{Left}{left}\SetKwData{This}{this}\SetKwData{Up}{up}
\SetKwInOut{Input}{input}\SetKwInOut{Output}{output}\SetKwInOut{Initialize}{Initialize:}
\Input{An ordered collection, $S$, of $n$ strictly increasing positive integers where $N:= s_{n-1} =\max(S)$}
\Output{An ordered collection, $P$, of the $m$ IMAPs in $S$ grouped by shared common difference, $d$. Each $p_0\in P$ is a triple, $\langle d, s, l \rangle$, containing the common difference, starting $S$ element, and length minus 1, respectively}
$P\leftarrow \langle\rangle\oplus \langle\rangle, i \leftarrow 0,1,\ldots,\lfloor\frac{N-S[0]}{2}\rfloor$ \Comment*[r]{$\langle\langle\langle d,s,l\rangle,\ldots\rangle,\ldots\rangle$}
$O \leftarrow \langle\rangle\oplus 1$ {\bf if} $i \in S$ {\bf else} $0$, $i \leftarrow 0,1,\ldots,N$\Comment*[r]{$\langle 0, 1, \ldots,1\rangle$}
$d\leftarrow 1$\Comment*[r]{Common difference}
\While(\Comment*[f]{No $d$ beyond midpoint of $S$}){$d \leq \lfloor\frac{N-S[0]}{2}\rfloor$}{
    $f_s \leftarrow \factors(d)$\Comment*[r]{All factors of $d$}
    \For(\Comment*[f]{Possible starting term of an IMAP}){$s \leftarrow S[0]$ {\bf to} $N$}{
        $l \leftarrow \extend(d, s, O)$\Comment*[r]{Determine maximal length (minus 1) for valid s and d}
        \If(\Comment*[f]{MAP of at least length 3 (minus 1)}){$l \geq 2$}{
            \If(\Comment*[f]{See lines 17--26}){$\textnormal{\isIMAP}$($d, s, l, f_s, P$)}{
                $P[d]\leftarrow P[d] \oplus \langle d,s,l\rangle$\Comment*[r]{New IMAP found}
                %\Comment*[r]{Computer Inner Metric Analysis; $O(A[|A|-1])$}
            }
        }
    }
    $d\leftarrow d+1$\Comment*[r]{Next $d$ for $s$ (not necessarily in $S$)}
}
\Return $P$
\;
 \vspace{1em}
\SetKwProg{Fn}{Function}{}{end}
  \Fn{\isIMAP($d, s, l, f_s, P$)}{
        \ForEach(\Comment*[f]{Each factor of $d$}){$d^{\prime} \in f_s$}{
                \ForEach(\Comment*[f]{Each IMAP with same common difference $d'$}){$p_0 \in P[d^{\prime}]$}{
                    %\If{$s<p_0.s$ {\bf or} $(s+l'\cdot d)>(p_0.s+p_0.l'\cdot p_0.d)$ {\bf or} $s\pmod{p_0.d} \not\equiv p_0.s\pmod{p_0.d}$}{
                    \If{{\bf not} $(s<p_0.s$ {\bf or} $(s+l\cdot d)>(p_0.s+p_0.l\cdot p_0.d)$ {\bf or} $s\pmod{p_0.d} \not\equiv p_0.s\pmod{p_0.d})$}{
                        \Return $false$ \Comment*[r]{MAP is not an IMAP}
                    }
                }
            }
        \Return $true$
  }
\caption[]{Enumerates inclusion-maximal arithmetic progressions (IMAPs) as implemented in \cite{Kranenburg2013} and originally in unpublished form by Chris Dyer (2011). Note that this is a reduced version of the original algorithm intended for computing an Inner Metric Analysis \cite{Nestke2001,Volk2008} of music, presented here in a form that better aligns with the central problem presented in this paper.}\label{alg:existing-algorithm}
\end{algorithm}

Alg.~\ref{alg:existing-algorithm} shows a simplified version of the procedural algorithm noted above \cite{Kranenburg2013}, shown here in a way which preserves its fundamental design while better illustrating its relevance to the central problem of this paper, namely, enumerating IMAPs. The algorithm takes as input an ordered collection, $S$, of strictly increasing positive integers and returns an ordered collection, $P$ (lines 1 and 16), of all IMAPs found in $S$. Each IMAP is uniquely defined as a triple, $\langle d,s,l\rangle$, where $d$ is its common difference, $s$ is its starting $S$ element, and $l$ is its length minus 1. 
The variable, $P$, is first initialized to be an ordered collection containing $\lfloor\frac{N-S[0]}{2}\rfloor$ empty ordered collections in line 1 followed by the helper variable, $O$, which is a boolean ordered collection containing, for each integer from $S[0]$ to $N$, one at the index corresponding to its value if it exists as an element in $S$ and zero if not (line 2). Alg.\ref{alg:existing-algorithm} works in a top-down fashion by first considering all positive common difference integers possible in $S$ (whether they are actually present or not), given $N$ and $k$ i.e., from 1 to $\lfloor\frac{N-S[0]}{2}\rfloor$ (lines 4--15). Then, for each such common difference, $d$, it computes all its factors ($\factors()$; line 5) and iterates once in lines 6--13 through all $s$ integers from $S[0]$ to $N$. For each $s$ in this loop, the function, $\extend()$, finds the maximal length (minus 1), $l$, of any arithmetic progression in $S$ beginning at $s$ if $s\in S$ and either $s < d$ or $s-d \notin S$: otherwise, returning a length of zero (line 7). If $l$ is at least 2, the MAP is checked in the $\isIMAP()$ function against all previously found IMAPs for each common difference factor, $d^{\prime}$ (lines 18--24), from the set of all factors computed earlier for $d$ (lines 19--23). Any MAP is not contained within any one IMAP, $p$, having the same common difference factor, $d^{\prime}$, if any one of the following conditions is true: (1) its leftmost element, $s$, is further left than the leftmost element of $p$, (2) its length extends beyond the length of $p$, or (3) in the case that $d=d^{\prime}$, it is disjoint with $p$ (i.e., belongs to a distinct residue class modulo $d$) (line 20; first, second, and third conditional expression, respectively). The algorithm terminates in line 16 by returning all discovered IMAPS stored in $P$.

%%%%%%%%%%%%%%%%
\begin{theorem}[Correctness]\label{theorem:correctness-existing}
For every strictly increasing sequence of $n$ positive integers, $S$, Alg.~\ref{alg:existing-algorithm} enumerates exactly the collection of all IMAPs of length, $k>2$, contained in $S$.
\end{theorem}

\begin{pf}
We rely on two facts that follow directly from the definition of an IMAP: (i) each IMAP is uniquely identified by the pair, $(s,d)$, consisting of its starting element and common difference, respectively, and (ii) if a MAP with a common difference, $d$, is properly contained in another MAP, then the containing MAP has a common difference, $d'$, with $d' \mid d$ and $d' < d$.
 
We establish the following invariant for the outer while loop (lines 4--15). At the start of the iteration with value, $d$, the collection, $\bigcup_{d' < d} P[d']$, contains exactly those IMAPs of $S$ whose common difference is less than $d$. Before the first iteration ($d=1$), $P$ is empty, so the invariant holds trivially. Next, we fix $d$ and assume the invariant holds. The inner for loop (lines 6--13) iterates $s$ over every integer in $[S[0], N]$, visiting every possible starting value. As implemented, $\extend()$ returns the maximum, $l$, such that $s, s+d, \ldots, s + l \cdot d \in S$ whenever $s \in S$ and either $s < d$ or $s-d \notin S$, and returns $0$ otherwise. Therefore, the conditional expression, $l \geq 2$, in line 8 is satisfied if and only if $(s, d)$ defines a MAP of length, $k > 2$, in $S$. Each candidate MAP is then assessed by $\isIMAP()$ (lines 18--24), which inspects every stored IMAP, $p_0 \in P[d']$, for each $d' \in \factors(d)$. The compound conditional expression in line 20 holds precisely when $p_0$ properly contains $(s, d)$, that is, $p_0$ has a starting element no later, an ending element no earlier, and when $d' = d$, their starting elements lie in the same residue class modulo $d'$. If $\isIMAP()$ rejects $(s, d)$, then some stored IMAP properly contains it, so $(s,d)$ is not an IMAP. However, suppose now that $(s, d)$ is not an IMAP. By (ii), some MAP, $(s_1, d_1)$, properly contains it with $d_1 \mid d$ and $d_1 < d$, so following this reasoning for the strictly decreasing divisors of $d$ will terminate at an IMAP, $(s^*, d^*)$, with $d^* \mid d$ and $d^*<d$, which by the invariant lies in $P[d^*]$ and, by transitivity, properly contains $(s,d)$. Since $d^* \in \factors(d)$, $\isIMAP()$ assesses $(s^*, d^*)$ and correctly rejects $(s,d)$. Therefore, $P[d]$ will contain exactly the IMAPs in $S$ with common difference, $d$, and the invariant is preserved.
 
No IMAP has a common difference exceeding $\lfloor\frac{N-S[0]}{2}\rfloor$, since such a progression would have fewer than three elements. Therefore, the outer loop halts once $d$ exceeds this bound and all inner loops range over finite collections, so the algorithm terminates.
 
By the invariant at termination, $\bigcup_d P[d]$ is exactly the collection of IMAPs of length, $k > 2$, in $S$.\qed
\end{pf}

%%%%%%%%%%%%%
\begin{theorem}[Time complexity]\label{alg-1-thm}
Let $S$ be a strictly increasing sequence of $n$ positive integers and $N$ be the value of its largest integer where $n\leq N$. Then, the time complexity of Alg.~\ref{alg:existing-algorithm} is $\bigO\left(N^{2+\frac{C}{\log\log N}}n\right)=\bigO(N^{2+o(1)}n)$ for some positive constant, $C$.
\end{theorem}

Before proving Theorem~\ref{alg-1-thm}, we first prove an auxiliary lemma on the number of IMAPs having the same common difference.

\begin{lemma}[Fixed-difference IMAP bound]\label{auxil-lem}
    For any common difference, $d$, the maximum number of IMAPs with $d$ in a sequence, $S$, of $n$ elements is $\lfloor \frac{n}{3} \rfloor$.
\end{lemma}

\begin{proofof}{Lemma~\ref{auxil-lem}}     
    If two MAPs in $S$ have the same common difference, $d$, they are by definition disjoint. Since each MAP has at least $k=3$ elements, there are at most $\frac{n}{3}$ MAPs that have the same common difference. Conversely, consider the sequence, $S = S(n,d)$, such that $S(n,d) = (d, 2d, 3d, 5d, 6d, 7d, 9d, \ldots)$, or more formally, $s_i = d \left( i + \left\lfloor \frac{ i - 1 }{ 3 } \right\rfloor \right)$. Then, $S$ has $\lfloor \frac{n}{3} \rfloor$ IMAPs of common difference, $d$, since all IMAPs are MAPs.\qed
\end{proofof}

\begin{proofof}{Theorem~\ref{alg-1-thm}}
The total work of the algorithm is carried out in the six loop iterations detailed below and presented in scope order from top to bottom.
\begin{enumerate}
    \item The outer while loop (lines 4--15) traverses all common difference integers, $d$, from 1 to the largest possible, given $N$ and the minimum value of $k=3$, which is $\bigO(N)$.
    \item As implemented, $\factors()$ (line 5) computes the factors of $d$ by traversing all integers from $1$ to $d$, which is $\bigO(N)$. It should be noted, however, that computing all factors of $d$ can be done more efficiently using the $\bigO(\sqrt d \log d) = \tilde\bigO(\sqrt d) = \tilde\bigO(\sqrt N)$ trial division algorithm.
    \item The loop in lines 6--13 traverses all integers from the smallest $S$ element to $N$, which is at most $\bigO(N)$.
    \item As implemented, $\extend()$ (line 7) traverses at most $n-1$ integers in $S$ from $s$ to $N$ by an increment of $d$, which is $\bigO(n)$, since $n\leq N$ and it will only ever be the case that $N$ integers are traversed when $n=N$.
    \item The for-each loop in lines 18--24 traverses all previously computed factors of $d$ from (2), which has been shown to be $\bigO(f(N)) = \bigO\left(N^{ \frac{C}{\log\log N} }\right)$, for some small positive constant, $C$, 
    where $f(N)$ denotes the number of divisors of $N$ \citep{Nicolas1983,Hardy2008}.
    \item The final nested for-each loop in lines 19--23 traverses all IMAPs having the same $d$, which is $\bigO(n)$ by Lemma \ref{auxil-lem}.
\end{enumerate}
From the above, the total time of Alg.~\ref{alg:existing-algorithm} can be expressed as
\[
    T_1\left(n,N\right) = \bigO \left( N\cdot \left(N+ \left(N\cdot\left(n+\left(N^{\frac{C}{\log\log N}}\cdot n\right)\right)\right)\right)\right) = \bigO \left( N^{2+\frac{C}{\log\log N}}n \right)
\]
which we can simplify to $\bigO(N^{2+o(1)}n)$.\qed
\end{proofof}

\begin{corollary}[Polynomial-growth regime]
If $N=\Theta\left(n^c\right)$ for some constant, $c\geq1$, then the time complexity of Alg.~\ref{alg:existing-algorithm} is $\bigO(n^{2c+1+o(1)})$.
\end{corollary}

\begin{rmk}[Space complexity]\label{rem:existing-algorithm}
Since Alg.~\ref{alg:existing-algorithm} must maintain a record of all $m$ IMAPs, the order of growth for which we show in section~\ref{sec:expected-value-bounds} to be asymptotically quadratic in $n$, the overall space complexity is thus dominated by either $m$ or the ordered collection, $O$, of indicator values for $S$, in the case that $n^2 \ll N$, and is therefore $\bigO(n^2+N)$.
%it is relatively straightforward to see that 
\end{rmk}

%%%%%%%%%%%%%%%%%%%%%%%%%%%%%%%%%%%%%%%%%%%%%%
%%%%%%%%%%% Start of contribution %%%%%%%%%%%%
%%%%%%%%%%%%%%%%%%%%%%%%%%%%%%%%%%%%%%%%%%%%%%
\section{Proposed Algorithm}\label{sec:proposed-algorithm}

\begin{algorithm}[!t]
%\small
\DontPrintSemicolon
\SetAlgoLined
\SetProcNameSty{textsc}
\SetProcArgSty{textsc}
\renewcommand{\SetProgSty}[1]{\renewcommand{\ProgSty}[1]{\textnormal{\csname#1\endcsname{##1}}\unskip}}%
\SetProgSty{textsc}
\SetKwData{Left}{left}\SetKwData{This}{this}\SetKwData{Up}{up}
\SetKwInOut{Input}{input}\SetKwInOut{Output}{output}\SetKwInOut{Initialize}{Initialize:}
\Input{An ordered collection, $S$, of $n$ strictly increasing positive integers where $N:= s_{n-1} =\max(S)$}
\Output{An ordered collection, $P^\prime$, of the $m$ IMAPs in $S$. Each $p_0\in P^\prime$ is a triple, $\langle s, d, e\rangle$, containing its starting $S$ element, common difference, and ending $S$ element, respectively.}
$P^\prime \leftarrow \langle\rangle$ \Comment*[r]{$\langle\langle s, d, e \rangle,\dots\rangle$}
$O \leftarrow \langle\rangle\oplus 1$ {\bf if} $i \in S$ {\bf else} $0$, $i \leftarrow 0,1,\ldots,N$\Comment*[r]{$\langle 0, 1, \ldots,1\rangle$}
\nlnonumber
\nlast $L, F, d_{max} \leftarrow \longprogressions(S)$\Comment*[r]{Modified from \cite{Erickson1999} and shown in Alg.~\ref{alg:dynamic-programming-algorithm}}
\nlnonumber
\nlast $F \leftarrow \factorsprime(F, d_{max})$\Comment*[r]{Modified Euler Linear Sieve shown in Alg.~\ref{alg:factors}}
\For(\Comment*[f]{Last two elements less than length 3}){$i \leftarrow 0$ \KwTo $n-3$}{
    $j \leftarrow i + 1$\;
    $s \leftarrow S[i]$\Comment*[r]{Starting term of a possible IMAP}
    $d \leftarrow S[j] - S[i]$\;
    $d_{max'} \leftarrow \min\left(\lfloor\frac{N-s}{2}\rfloor, d_{max}\right)$\Comment*[r]{Greatest possible remaining common difference}
    $\ell \leftarrow \langle\rangle\oplus 0, i \leftarrow 0, 1, \ldots, d_{max}$\Comment*[r]{AP lengths associated to respective $d$}
     \While(\Comment*[f]{No $d$ beyond remaining midpoint}){$d \leq d_{max'}$ {\bf and} $j < n - 1$} {
        \If(){$L[i][j] > 2$} {
            $\ell[d] \leftarrow L[i][j]$\Comment*[r]{AP length for this $d$}
            \If(\Comment*[f]{Start of a new MAP}){$s - d < 0$ {\bf or} {\bf not} $O[s - d]$} {
                $e \leftarrow s + d \cdot (\ell[d] - 1)$\Comment*[r]{Ending term of a possible IMAP}
                $f_s' \leftarrow F[d]$\Comment*[r]{All $p\mid d$ co-factors of $d$} 
                \nlnonumber
                \nlast\If(\Comment*[f]{See lines 27--37}){$\textnormal{\isIMAPprime}(s, \ell, e, f_s')$} {
                    $P^\prime\leftarrow P^\prime \oplus \langle s,d,e\rangle$\Comment*[r]{New IMAP found}
                    }
                }
            }
            $j\leftarrow j + 1$\;
            $d \leftarrow S[j] - S[i]$\Comment*[r]{Next $d$ in $S$ for $s$}
        }
    }
\Return $P^\prime$\;
\nlnonumber
\vspace{1em}
\SetKwProg{Fn}{Function}{}{end}
\nlast  \Fn{\isIMAPprime($s, \ell, e, f_s'$)}{
    \ForEach(\Comment*[f]{Each $p\mid d$ co-factor of $d$}){$d' \in f_s'$}{
        \If(\Comment*[f]{Common difference factor $d'$ exists for $s$}){$\ell[d'] > 0$} {
           $e' \leftarrow s + d' \cdot (\ell[d'] - 1)$\Comment*[r]{End term of AP with factor $d'$}
           \If(\Comment*[f]{MAP with $e$ contained in AP with $e'$}){$e \leq e'$} {
                \Return $false$\Comment*[r]{Not an IMAP}
            }
        }
    }
    \Return $true$
  }
\caption{Enumerates inclusion-maximal arithmetic progressions (IMAPs) based in part on the algorithm offered in \cite{Erickson1999} for finding the longest arithmetic progression. Note that asterisks denote modified versions of existing procedures.}\label{alg:proposed-algorithm}
\end{algorithm}

As an introduction to our algorithm shown in Alg.~\ref{alg:proposed-algorithm}, the general procedure amounts to calling Alg.~\ref{alg:dynamic-programming-algorithm}, which is a modified version of the bottom-up dynamic programming algorithm for finding the single longest arithmetic progression \cite{Erickson1999} shown in section~\ref{sec:longest-arithmetic-progression-alg}, and efficiently extracting all IMAPs from the lookup table it constructs. We should note, however, that because we are concerned only with a constant minimal value of $k>2$, the improved asymptotic complexity of the divide-and-conquer approach for finding all MAPs \cite{Erickson1999} over the $\bigO(n^2)$ approach for minimal values of $k > \log \, n \, \log \, \log \, n$ (noted in section~\ref{sec:related}), means that the latter approach---which forms the basis of the proposed algorithm---will be more efficient for all but the smallest $n$. Were we to consider instead a generalization of the problem in which the minimal of $k$ could vary and increase in size relative to $n$, the divide-and-conquer approach would indeed be more efficient for discovering all MAPs for these certain values of $k$. It is not immediately clear, however, whether this approach could be efficiently modified to enumerate all IMAPs. 

Finally, the proposed algorithm relies on Lemma~\ref{lemma:imap-characterization-via-primes} below to re-characterize inclusion-maximality as determined by Alg.~\ref{alg:existing-algorithm} which, combined with its bottom-up approach, allows for a significant improvement in overall efficiency.

\begin{lemma}[IMAP characterization via primes]\label{lemma:imap-characterization-via-primes}
    A MAP with common difference, $d$, is inclusion-maximal if and only if it is not properly contained in any other MAP with common difference, $c = \frac{d}{p}$, for any prime, $p\mid d$.
\end{lemma}
\begin{pf}
    Let $A = \langle a_0, a_1 = a_0 + d, \dots, a_{k-1} = a_0 + (k-1)d \rangle$ be a MAP that is not an IMAP. Therefore, there is a MAP $A' = \langle a'_0 = a_0, \ldots, a'_{k'-1} = a_0 + (k'-1) d' \rangle$ with common difference, $d'<d$, that properly contains $A$. Noting that $d'\mid d$, let $p$ be any prime factor of $\frac{d}{d'}$. Then, the arithmetic progression, $B = \langle b_0 = a_0, b_1 = a_0 + \frac{d}{p}, \dots, b_{p(k-1)} = a_0 + (k-1)d = a_{k-1} \rangle$, has common difference, $\frac{d}{p}$, and properly contains $A$.\qed
\end{pf}

\noindent As such, for any integer, $d$, and any prime factor, $p$ of $d$, we refer to $c = \frac{ d }{ p }$ as the $p \mid d$ co-factor of $d$.

%%%%%%%%%%%
\subsection{Algorithm 2}

Alg.~\ref{alg:proposed-algorithm} takes as input an ordered collection, $S$, of strictly increasing positive integers and returns an unordered collection, $P^\prime$ (lines 1 and 23), of all IMAPs found in $S$. Each IMAP in $P^\prime$ is a triple, $\langle s, d, e \rangle$, corresponding to its starting $S$ element, $s$, common difference, $d$, and ending $S$ element, $e$. The variable, $P^\prime$, is first initialized to be empty in line 1 followed by a helper variable, $O$ (line 2), which serves the same purpose as it did in Alg.~\ref{alg:existing-algorithm}. 
 
The primary work by the algorithm is carried out by $\longprogressions()$ \cite{Erickson1999} (Alg.~\ref{alg:dynamic-programming-algorithm}) called in line 3 and within the for loop in lines 5--25. In its original presentation, $\longprogressions()$ returns the length of the single longest arithmetic progression. However, in the course of finding the length of this progression, it maintains a lookup table, $L$, that stores for each $(i<j)$-pair in $S$, the maximal length of the arithmetic progression having as its first two terms the elements in $S$ indexed by this pair. Importantly, we have made three modifications to this existing algorithm: (1) we return the lookup table, $L$, rather than the length of the longest arithmetic progression while removing the few components necessary for maintaining this length, (2) use its traversal through all $(i<j)$ pairs in $S$ that form an arithmetic triple to initialize and then return an unordered collection of key-value pairs, $F$, containing as each key a distinct common difference, and (3) compute and return the largest common difference key, $d_{max}\in F$. In line 4, we employ a procedure in $\factorsprime()$, based on a modified version of the Euler's Linear Sieve algorithm (Alg.~\ref{alg:factors}), for computing all necessary $p\mid d$ co-factors as defined in Lemma~\ref{lemma:imap-characterization-via-primes}. For this procedure, we use the Euler Linear Sieve algorithm to compute the standard sieve containing the smallest prime factors for the integers from 2 to $d_{max}$ and then compute and store the collection of all such $p\mid d$ co-factors for each distinct common difference, $d\in F$, found by $\longprogressions()$ \cite{Erickson1999} and return the completed $F$. Further detail regarding Alg.~\ref{alg:dynamic-programming-algorithm} and Alg.~\ref{alg:factors} is provided in sections~\ref{sec:longest-arithmetic-progression-alg} and \ref{sec:factors-alg}, respectively.

The other main procedure of the algorithm in lines 5--25 traverses the upper-right triangle of $L$ by row and column and efficiently checks, for each corresponding $(i<j)$ pair of $S$ elements for a given $i$, whether the maximal length of the arithmetic progression when taking this pair of elements to be its first two terms (regardless of whether these are its true starting terms) is contained within any other arithmetic progression having a common difference that is a co-factor of the form $p\mid d$ with smaller $j$. An outer for loop begins in line 5 by iterating from the $i$th row of $L$ and $i$th element of $S$ from 0 to $n-3$. Inside this loop, $j$ is set equal to $i+1$ (line 6), the starting term of a possible IMAP, $s$, is set equal to $S[i]$ (line 7), and $d$ is set equal to the common difference for the current $(i<j)$ pair (line 8). For each $i$th row in $L$, only $j$ corresponding to the $(i<j)$ pairs having a common difference less than or equal to $\min\left(\lfloor\frac{N-s}{2}\rfloor, d_{max}\right)$ for the current $s$ need to be considered, so this maximum $d$ value, $d_{max'}$, is initialized in line 9. A helper variable, $\ell$, is then initialized in line 10 and associates the maximal length of each $(i<j)$ arithmetic progression for this $i$ that could contain any potential IMAP also beginning from this $i$, to each respective common difference so far discovered up to $j$. A while loop then begins in lines 11--24 for each $d$ less than or equal to $d_{max'}$. Inside this while loop, the maximal length of each $(i<j)$ arithmetic progression having a length $k>2$ is stored at an index in $\ell$ equal to its common difference (line 13). A conditional expression in line 14 then determines whether this $(i<j)$ arithmetic progression is the start of a new arithmetic progression, and if so, is a candidate for an IMAP. The function, $\isIMAPprime()$ (line 17 and 27--37), then loops through all previously computed $p\mid d$ co-factors, $d'$, of the common difference, $d$, of the given $(i<j)$ arithmetic progression and checks whether any are so far present in $\ell$ for this $i$ by $\ell[d']>0$ (line 29). If such a $p\mid d$ co-factor has been found for this $i$, it then checks whether the maximal length of its arithmetic progression corresponds to an ending element that is at least as great at the ending element of the candidate $(i<j)$ arithmetic progression (line 31), and if so, the candidate $(i,j)$ arithmetic progression is not an IMAP and the $\isIMAPprime()$ returns \textit{false}. If none of the $p\mid d$ factors of the candidate $(i,j)$ arithmetic progression's common difference exist in $\ell$ and have an ending element at least as great, then an IMAP has been found and is stored in $P^\prime$ (line 18). In lines 22--23, $j$ is first incremented by 1 and then the next common difference, $d$, is set equal to the difference between the new $S[j]$ element and $S[i]$. After this while loop concludes for all $i$, the algorithm returns in line 26 with all discovered IMAPs stored in $P^\prime$.

Finally, it is worth highlighting some further differences between Alg.~\ref{alg:existing-algorithm} and Alg.~\ref{alg:proposed-algorithm} as well as possible improvements that can be made to the latter. For example, if one were concerned only with counting the number of IMAPs, the relatively high cost of appending to an ordered collection in line 18 of Alg.~\ref{alg:proposed-algorithm} can be avoided by simply not storing each IMAP. In contrast, this is not possible with Alg.~\ref{alg:existing-algorithm} since every one potential IMAP must be checked against a record containing already discovered IMAPs. If one does wish to retrieve the IMAPs in Alg.~\ref{alg:proposed-algorithm}, however, this can be done slightly more efficiently in terms of absolute space, than presented above by simply modifying the lookup table, $L$. As $L$ is traversed in lines 5--25, it can be efficiently updated so as to zero out all $(i<j)$ corresponding to non-IMAP locations with the final result being that any remaining values greater than 2 are IMAPs. The collection of all IMAPs can then be retrieved in $\bigO(n^2)$ by again traversing $L$ and halting in each $i$th row when a 2 is reached. Perhaps most importantly, however, the proposed algorithm can be improved further, even when $k>2$ and thus without using the divide-and-conquer approach noted above, if the number of MAPs is small. By collecting all $m'$ MAPs found in $\longprogressions()$, all IMAPs can be filtered directly from these without having to traverse $L$ again in lines 5--25. If their number is small, e.g., $m'=\bigO(n)$, filtering in this way can be more efficient (see Remark~\ref{rmk:few-maps}). 

%%%%%%%%%%%%%%%%
\begin{theorem}[Correctness]\label{theorem:correctness-proposed}
For every strictly increasing sequence of $n$ positive integers, $S$, Alg.~\ref{alg:proposed-algorithm} enumerates exactly the collection of all IMAPs of length, $k>2$, contained in $S$.
\end{theorem}

\begin{pf} 
We rely on the correctness of Algs.~\ref{alg:dynamic-programming-algorithm} and \ref{alg:factors} (lines 3--4) proven in sections~\ref{sec:longest-arithmetic-progression-alg} and \ref{sec:factors-alg}, respectively. By the correctness of Alg.~\ref{alg:dynamic-programming-algorithm}, $L[i][j]$ is the length of the longest arithmetic progression in $S$ whose first two terms are $S[i]$ and $S[j]$, and by the correctness of Alg.~\ref{alg:factors}, $F[d]$ is the collection of $p \mid d$ co-factors of $d$ for every distinct prime factor, $p$ of $d$. 

We establish the following invariant for the outer for loop in lines 5--25: At the start of the iteration with index, $i$, the collection, $P^\prime$, contains exactly those IMAPs in $S$ whose starting element is $S[i']$ for some $i' < i$. Before the first iteration ($i = 0$), $P^\prime$ is empty, so the invariant holds trivially. Next, we fix $i$, set $s = S[i]$, and assume the invariant. The inner while loop in lines 11--24 increments $j$ monotonically, visiting each pair, $(i, j)$, with $d := S[j] - S[i] \leq d_{max'} := \min(\lfloor\frac{N-s}{2}\rfloor, d_{max})$ exactly once. The bound, $d_{max'}$, does not exclude any possible candidate as any MAP of length, $k > 2$, starting at $s$ has $d \leq \lfloor\frac{N-s}{2}\rfloor$, and any progression of length, $\geq 3$, in $S$ has $d \leq d_{max}$ by definition of $d_{max}$. Lines 12 and 14 enforce a maximal length, $\geq 3$ with $L[i][j] > 2$ and left-maximality with $s - d \notin S$, respectively, while the right endpoint, $e = s + d(L[i][j] - 1)$, is computed in line 15. Therefore, $\isIMAPprime()$ in line 17 assesses exactly the MAPs in $S$ starting at $s$ and having $d$ and $e$.

We now show that $\isIMAPprime()$ accepts $(s, d)$ if and only if it is an IMAP. Suppose first that $\isIMAPprime()$ returns \textit{false}. Then, for some $d' \in F[d]$, $\ell[d'] > 0$ and $e' := s + d'(\ell[d'] - 1) \geq e$. By line 13, $\ell[d']$ equals $L[i][j']$ for the unique index, $j'$, with $S[j'] = s + d'$. Therefore, $s, s + d', \ldots, e'$ all lie in $S$. Since $d' \mid d$ with $d' < d$, every element belonging to the candidate, $(s,d)$, is an element of this longer progression that properly contains it. Extending this longer progression leftward in $S$ as far as possible yields a MAP in $S$ that still properly contains $(s, d)$, so $(s, d)$ is not an IMAP. Suppose now that $(s, d)$ is not an IMAP. By Lemma~\ref{lemma:imap-characterization-via-primes}, there is a MAP, $B$, with common difference, $d' \in F[d]$, properly containing it. Even when $B$ starts at an element strictly less than $s$, the inclusion, $\{s, s + d, \ldots, e\} \subseteq B$, ensures that $\{s, s + d', s + 2d', \ldots, e\} \subseteq B \subseteq S$, so the longest progression in $S$ starting at $(s, s + d')$ has a rightmost element, $e^* \geq e$, and length, $\geq \frac{e - s}{d'} + 1 \geq 3$. Since $d' < d$, the index, $j'$, with $S[j'] = s + d'$ satisfies $j' < j$, so this progression was visited earlier in the same iteration of $i$, and line 13 sets $\ell[d']$ accordingly. Therefore, $\isIMAPprime()$ returns \textit{false}. Since exactly the IMAPs starting at $s$ are stored in $P^\prime$ during iteration, $i$, the invariant is preserved.

The outer for loop ranges over the finite index collection, $\{0, 1, \ldots, n - 3\}$, the inner while loop strictly increments $j$, which is bounded by $n - 1$, and $\isIMAPprime()$ iterates over the finite collection, $F[d]$, so the algorithm terminates. 

By induction on $i$, after the outer for loop terminates, $P^\prime$ contains exactly the IMAPs of length, $k > 2$, in $S$.\qed
\end{pf}

%%%%%%%%%%%%%%%%
\begin{theorem}[Time complexity]\label{theorem:complexity-proposed}
Let $S$ be a strictly increasing sequence of $n$ positive integers and $N$ be the value of its largest integer where $n\leq N$. Then, the time complexity of Alg.~\ref{alg:proposed-algorithm} is $\bigO\left( n^2 \frac{ \log N }{ \log \log N } + N \right)$.
\end{theorem}

\begin{pf}
The total work by the algorithm is carried out in the procedures detailed below and presented in scope order from top to bottom.
\begin{enumerate}
    \item The modified $\longprogressions()$ algorithm, Alg.~\ref{alg:dynamic-programming-algorithm}, originally from \cite{Erickson1999} and called in line 3 has a complexity of $\bigO(n^2)$. The proof is provided in section~\ref{sec:longest-arithmetic-progression-alg}.
    \item \sloppy The $\factorsprime()$ procedure, Alg.~\ref{alg:factors}, based on Euler's Linear Sieve algorithm and called in line 4 is $\bigO \left( N + \min \left( n^2 \frac{ \log N }{ \log \log N } , N \log\log N \right) \right)$. The proof is provided in section~\ref{sec:factors-alg}.
    \item The outer for loop in lines 5--25 traverses $n-3$ elements of $S$ once, which is $\bigO(n)$.
    \item The nested while loop in lines 11--24 traverses at most $n-1$ elements of $S$, which is $\bigO(n)$.
    \item Finally, the nested for-each loop in $\isIMAPprime$ (lines 17 and 27--37) traverses, for a given $d$, all $p\mid d$ co-factors (for all $p$ where $p$ is a distinct prime factor of $d$) previously computed from (2). This is equal to the number of prime factors of $d$, which is $\bigO\left(\frac{ \log d }{ \log \log d }\right)$ in general. Since $d \le N$, we obtain a runtime of  $\bigO\left(\frac{ \log N }{ \log \log N }\right)$.
\end{enumerate}
\sloppy From the above, the total time of Alg.~\ref{alg:proposed-algorithm} can be expressed as
\[
    T_2( n, N ) = \bigO \left( n^2 + \left( N + \min \left( n^2 \frac{ \log N }{ \log \log N } , N \log\log N \right)\right) + n^2\frac{ \log N }{ \log \log N } \right) = \bigO\left( n^2 \frac{ \log N }{ \log \log N } + N \right).
\]

\noindent which we can simplify to $\bigO(n^2N^{o(1)}+N)$. While $n^2 \frac{ \log N }{ \log \log N }$ dominates $N$ whenever $n=\sqrt{N\frac{\log\log N}{\log N}}$, without a bound on $N$ with respect to $n$, either term may dominate the running time over the other depending on the specific values of $n$ and $N$. Therefore, both terms are retained.\qed
\end{pf}

\begin{theorem}[Time complexity refined]\label{theorem:complexity-proposed-refined}
\sloppy Let $S$ be a strictly increasing sequence of $n$ positive integers, and $N$ be the value of its largest integer with $n \le N$. Then, the refined time complexity of Alg.~\ref{alg:proposed-algorithm} is $\bigO\left(\min\left( n^2 \frac{\log N}{\log\log N}, nN\log \log N \right) + N \right)$.
\end{theorem}

\begin{pf}
By the proof of Theorem~\ref{theorem:correctness-proposed}, $\isIMAPprime()$ is called exactly once for each pair, $(s, d)$, that begins a MAP of length, $k > 2$. The collection of all MAPs we will denote by $M$. For a fixed $(s,d)\in M$, $\isIMAPprime$ traverses the list of $p\mid d$ co-factors previously computed by $\factorsprime()$, the number of which is $\omega(d)$. Across all calls, the total cost is thus $\sum_{(s,d)\in M} \omega(d) = \sum_d c_d\,\omega(d)$, where $c_d$ denotes the number of MAPs in $S$ with common difference, $d$. We bound this sum in the following two ways.

\begin{enumerate}
    \item \sloppy Every $d\le N$ has $\omega(d) = \bigO\left(\frac{\log N}{\log \log N}\right)$, which is the maximal order of $\omega$ \citep{Hardy2008} for numbers having many prime factors (e.g., primorials). By Lemma~\ref{lemma:lower-bound-MAPs}, the total number of MAPs in $S$ is $\sum_d c_d \le \tfrac{1}{2}\binom{n}{2} = \bigO(n^2)$. Therefore, $\sum_{(s,d)\in M} \omega(d) = \bigO\left(n^2\frac{\log N}{\log \log N}\right)$.
    \item On the other hand, consider the function, $\psi(p) = |\{(s,d)\in M : p\mid d\}|$, and recall from the proof of Lemma~\ref{auxil-lem} that MAPs having the same common difference, $d$, are pairwise disjoint subsets of $S$ each of length $k > 2$, and their number is at most $\lfloor \frac{n}{3} \rfloor$. It follows that $c_d \le \lfloor \frac{n}{3} \rfloor$ and since all possible multiples of $p$ lie in $\{1,2,\dots,N\}$ and number at most $\lfloor \frac{N}{p} \rfloor$, then $\psi(p) \le \lfloor \frac{n}{3} \rfloor \cdot \lfloor \frac{N}{p} \rfloor \le \frac{nN}{3p}$. Therefore, we obtain
    \begin{align*}
        \sum_{(s,d)\in M} \omega(d) &= \sum_{(s,d)\in M} |\{(p,d) : p\text{ prime}, p\mid d\}| \\
        &= \sum_{p\le N} \psi(p) \\
        &\le \sum_{p\le N} \frac{nN}{3p} \\
        &= \frac{nN}{3} \sum_{p\le N} \frac{1}{p} \\
        &= \frac{nN}{3}\left(\log\log N + A + o\!\left(\frac{1}{\log N}\right)\right) \\
        &= \frac{nN}{3}\log \log N + \bigO(nN)
    \end{align*}
    by Mertens' second theorem (where $A$ is the Meissel--Mertens constant) \citep{Hardy1917,Hardy2008,Chen2025}, so $\sum_{(s,d)\in M} \omega(d) = \bigO(nN\log\log N)$.
\end{enumerate}
Taking the smaller of (1) and (2) leaves $\bigO\left(\min\left(n^2\frac{\log N}{\log \log N}, nN\log \log N\right)\right)$. Substituting this quantity for the third term in the full expression of $T_2(n,N)$ in Theorem~\ref{theorem:complexity-proposed} (i.e., procedures (3) to (5)) gives a total refined time for Alg.~\ref{alg:proposed-algorithm} of
\[
    T_2'(n,N) = \bigO\left(\min\left( n^2\frac{\log N}{\log\log N},\, nN\log\log N \right) + N \right).\hfill\qed
\]
\end{pf}

\begin{corollary}[Polynomial-growth regime]\label{sec:cor-complexity}
\sloppy If $N=\Theta\left(n^c\right)$ for some constant, $c\geq1$, then by Theorem~\ref{theorem:complexity-proposed}, $n^{2+o(1)}$ will dominate whenever $c\leq2$ while $n^c$ (i.e., $N$) will dominate whenever $c>2$ up to a poly-logarithmic factor. For $c=1$ specifically, Theorem~\ref{theorem:complexity-proposed-refined} tightens $n^{2+o(1)}$ to explicitly $\bigO(n^2 \log\log n)$. Without further information regarding $c$, the time complexity of Alg.~\ref{alg:proposed-algorithm} is $\bigO(\max(n^{2+o(1)}, n^c)) = \bigO(n^{\max(2+o(1), c)})$.
\end{corollary} 

Importantly, if we do consider a dense regime where $n$ grows linearly with $N$ (equivalently, $c=1$), the $\bigO(n^2 \log\log n)$ time complexity of Alg.~\ref{alg:proposed-algorithm} by Theorem~\ref{theorem:complexity-proposed-refined} is an improvement over Alg.~\ref{alg:existing-algorithm} that is asymptotically better than $n$.

\begin{proposition}[Lower bound]\label{prop:lower-bound}
Any algorithm that enumerates all IMAPs contained in a strictly increasing sequence, $S$, of $n$ positive integers must perform at least $\Omega(n^2)$ queries in the 3-linear decision tree model. Therefore, Alg.~\ref{alg:proposed-algorithm} is time-optimal up to a factor of $n^{o(1)}$ relative to this lower bound whenever $N=\Theta(n)$.
\end{proposition}
 
\begin{pf}
We reduce the problem of finding the single longest arithmetic progression, which is shown in \citep{Erickson1999} to require $\Omega(n^2)$ queries in the 3-linear decision tree model, to the problem of enumerating IMAPs. By definition, the longest arithmetic progression of length, $k>2$, in $S$ is an IMAP, so it appears in the output of any correct algorithm for enumerating IMAPs, where its length equals the maximum length over the collection of enumerated IMAPs (and equals 2 if no IMAPs exist). If an algorithm, $\mathcal{A}$, correctly enumerates all IMAPs, one may append a post-processing procedure that scans the output of $\mathcal{A}$ and returns the IMAP of maximum length. Since scanning requires reading only the already-produced output with no further queries to the input, $\mathcal{A}$ solves the longest arithmetic progression problem using exactly the same number of 3-linear queries. Combined then with Theorem~\ref{theorem:complexity-proposed}, when $N=\Theta(n)$, the running time of $\bigO(n^{2+o(1)})$ exceeds the $\Omega(n^2)$ decision-tree lower bound by only the factor, $\log\log n$, so Alg.~\ref{alg:proposed-algorithm} is time-optimal up to $\log\log n$.\qed
\end{pf}

\begin{rmk}[Filtering with few MAPs]\label{rmk:few-maps} When the number of MAPs, $m'$, grows only linearly with $n$ (i.e., $m' = O(n)$), a post-process filtering procedure can enumerate all IMAPs from the collection of MAPs produced by Alg.~\ref{alg:dynamic-programming-algorithm} in at most $\bigO(n^2 + m'^2)=\bigO(n^2)$ time. This avoids the $\bigO(n^2 \frac{ \log N }{ \log \log N } + N )$ dominating cost of Alg.~\ref{alg:proposed-algorithm}, recovering quadratic complexity and remaining time-optimal by Proposition~\ref{prop:lower-bound}. The exact cross-over point at which a brute-force pair-wise filtering of all MAPs becomes more efficient is $m'=\Theta\left(\max\left(n \sqrt{\frac{ \log N }{ \log \log N }}, \sqrt N\right)\right)$, however, we show in section~\ref{sec:expected-value-bounds} that the expected order of growth for IMAPs means filtering is rarely the more efficient approach. For example, when $N$ grows linearly with $n$, the expected number of IMAPs grows quadratically with $n$, which makes filtering not at all efficient in this regime. While we further show empirically in section~\ref{sec:experiments} that the number of IMAPs and thus MAPs decrease after some point as the density of the sequence decreases (i.e., as $c>1$ increases in $N=\Theta(n^c)$), a precise characterization of this relationship remains open.
\end{rmk}

\begin{rmk}[Space complexity]
The overall space complexity of Alg.~\ref{alg:proposed-algorithm} is dominated by the lookup table, $L$, the ordered collection, $O$, of indicator values for $S$, or the collection of $p\mid d$ co-factors for all discovered common differences (i.e., the variable $F$), and is therefore $\bigO\left(n^2+N+\min \left( n^2 \frac{ \log N }{ \log \log N } , N \log\log N \right)\right)$. In comparison to Alg.~\ref{alg:existing-algorithm}, which has a space complexity of $\bigO(n^2+N)$ (see Remark~\ref{rem:existing-algorithm}), the complexity of Alg.~\ref{alg:proposed-algorithm} is equivalent up to a poly-logarithmic factor (since its third term never exceeds $N\log\log N$), with both exactly equivalent at $\Theta(n^2)$ when $N = \Theta(n)$.
\end{rmk}

%%%%%%%%%%%
\subsection{Algorithm 3}\label{sec:longest-arithmetic-progression-alg}

Alg.~\ref{alg:dynamic-programming-algorithm} is called in line 3 of Alg.~\ref{alg:proposed-algorithm} and is a slightly modified version of the dynamic programming algorithm first introduced in \citep{Erickson1999} for finding the longest arithmetic progression. As originally presented, the correctness and time complexity of the algorithm were stated but not proved, so we provide proofs of both below for completeness.

\begin{algorithm}[!t]
\small
\DontPrintSemicolon
\SetProcNameSty{textsc}
\SetProcArgSty{textsc}
\renewcommand{\SetProgSty}[1]{\renewcommand{\ProgSty}[1]{\textnormal{\csname#1\endcsname{##1}}\unskip}}%
\SetProgSty{textsc}
\SetKwData{Left}{left}\SetKwData{This}{this}\SetKwData{Up}{up}
\SetKwInOut{Input}{input}\SetKwInOut{Output}{output}\SetKwInOut{Initialize}{Initialize:}
\Input{An ordered collection, $S$, of $n$ strictly increasing positive integers where $N:= s_{n-1} =\max(S)$}
\Output{A 2D ordered collection, $L$, containing for each $(i<j)$ pair in $S$, the longest length of an arithmetic progression beginning from this pair of elements, and an unordered key-value collection, $F$, with a key for each distinct common difference across all discovered arithmetic triples}
\SetKwProg{Fn}{Function}{}{end}
\Fn{\longprogressions($S$)}{
$L\leftarrow\langle\rangle\oplus(\langle\rangle\oplus 0, j \leftarrow 0,1,\ldots,n-1), i \leftarrow 0,1,\ldots,n-1$\Comment*[r]{$(i<j)$ AP lengths}
\nlnonumber
\nlast $F \leftarrow \lbrace\rbrace$\Comment*[r]{e.g., $\lbrace\langle 2 : \,\langle\rangle\rangle, \ldots\rbrace$}
\nlnonumber
\nlast $d_{max} \leftarrow 0$\Comment*[r]{Largest common difference}
\For(\Comment*[f]{O-based indexing}){$j \leftarrow n - 2$ \KwTo $1$}{
    $i\leftarrow j-1$\; %; k\leftarrow j+1
    $k\leftarrow j+1$\;
    \While{$i \geq 0$ {\bf and} $k < n$}{
        \uIf(\Comment*[f]{\tikz\draw[black,fill=black] (0,0) circle (.4ex);\quad\quad\tikz\draw[black,fill=black] (0,0) circle (.4ex);\quad\tikz\draw[black,fill=black] (0,0) circle (.4ex);}){$S[i] + S[k] < 2 \cdot S[j]$}{
            $k\leftarrow k+1$
        }
        \uElseIf(\Comment*[f]{\tikz\draw[black,fill=black] (0,0) circle (.4ex);\quad\tikz\draw[black,fill=black] (0,0) circle (.4ex);\quad\quad\tikz\draw[black,fill=black] (0,0) circle (.4ex);}){$S[i] + S[k] > 2 \cdot S[j]$}{
            $L[i][j] \leftarrow 2$\;
            $i\leftarrow i-1$
        }
        \Else(\Comment*[f]{AP triple found}){
            $L[i][j] \leftarrow \max(2, L[j][k])+1$\Comment*[r]{Update length}
            \nlnonumber
            \nlast\If(\Comment*[f]{New AP common difference found}){$S[j]-S[i] \notin F$}{
                \nlnonumber
                \nlast$F[S[j]-S[i]] \leftarrow \langle\rangle$
            }
            \nlnonumber
            \nlast\If(\Comment*[f]{Larger AP common difference found}){$S[j]-S[i] > d_{max}$}{
                \nlnonumber
                \nlast$d_{max} \leftarrow S[j]-S[i]$
            }
            $i\leftarrow i-1$\;
            $k\leftarrow k+1$
        }
    }
   \While{$i \geq 0$}{
        $L[i][j] \leftarrow 2$\;
        $i\leftarrow i-1$
   }
}
\nlnonumber
\nlast \Return $L, F, d_{max}$
}
\caption{Modified dynamic programming algorithm from \citep{Erickson1999} for finding the length of the single longest arithmetic progression in a strictly increasing integer sequence. Note that asterisks indicate modifications to the original algorithm as used to enumerate inclusion-maximal arithmetic progressions (IMAPs) with the proposed algorithm shown in Alg.~\ref{alg:proposed-algorithm}.}\label{alg:dynamic-programming-algorithm}
\end{algorithm}

%%%%%%%%%%%%%%%%
\begin{theorem}[Correctness]\label{theorem:correctness-dynamic-programming}
For every strictly increasing sequence of $n$ positive integers, $S$, Alg.~\ref{alg:dynamic-programming-algorithm} computes, for all $(i<j)$, the length, $L[i][j]$, of the longest arithmetic progression in $S$ beginning from $(s_i,s_j)$. Moreover, the algorithm computes the collection, $F$, of all distinct common differences belonging to arithmetic progressions triples in $S$, and the maximum such difference, $d_{max}$. 
\end{theorem}

\begin{pf}
We begin by defining $L^*(i, j)$ to be the length of the longest arithmetic progression in $S$ beginning at $S[i]$ and $S[j]$, with $L^*(i, j) = 2$ when no extension of length, $\geq 3$, exists.
 
We establish the following invariant. At the start of each iteration of the outer for loop in lines 5--30 with index, $j$ (processed in decreasing order from $n - 2$ to $1$): for every $j' \in \{j + 1, \ldots, n - 2\}$ and $i' < j'$, $L[i'][j'] = L^*(i', j')$, $F$ is the collection of common differences belonging to all arithmetic triples in $S$ with middle index in $\{j + 1, \ldots, n - 2\}$, and $d_{max} = \max(F \cup \{0\})$. Before the first iteration ($j = n - 2$), $F$ is empty and $d_{max} = 0$, so the invariant holds trivially. Next, we fix $j$ and assume the invariant. The two-pointer traversal in lines 8--25 starts at $i = j - 1$ and $k = j + 1$. The strict monotonicity of $S$ means, for each $i < j$, at most one $k > j$ with $S[i] + S[k] = 2 \cdot S[j]$, and the satisfying $S[k^*]$ strictly increases as $i$ decreases. We thus encounter the following three cases. When $S[i] + S[k] < 2 \cdot S[j]$, any satisfying, $k^*$, for $i$ means $k^* > k$, so incrementing $k$ does not miss any potential arithmetic triple. When $S[i] + S[k] > 2 \cdot S[j]$, no arithmetic triple for $i$ exists since $k$ only increases, and thus $L^*(i, j) = 2$ (line 12). When $S[i] + S[k] = 2 \cdot S[j]$, the terms, $S[i]$, $S[j]$, and $S[k]$, are an arithmetic triple, so $L^*(i, j) = L^*(j, k) + 1$. For $k \leq n - 2$, the invariant means $L[j][k] = L^*(j, k) \geq 2$, so in line 15 $L[i][j]$ becomes the larger of 2 or $L[j][k]$ then plus 1, which equals $L^*(i, j)$. For $k = n - 1$, $L[j][k] = 0$ (as column $n - 1$ is not modified) and $L^*(j, k) = 2$ since $S[n - 1] = \max(S)$, and thus $L[i][j] = 3$, which equals $L^*(j, k) + 1 = L^*(i, j)$. Lines 16--21 then store, $d = S[j] - S[i]$, in $F$ if it is absent and update $d_{max}$. If the outer for loop halts due to $k = n$, the second inner while loop in lines 26--29 sets $L[i'][j] = 2$ to all remaining $i'$. Since the current $i'$ satisfies $S[i'] + S[n - 1] < 2 \cdot S[j]$, the satisfying $S[k^*]$ equals $2 \cdot S[j] - S[i']$, which strictly exceeds $S[n - 1] = \max(S)$, and since $S[k^*]$ only increases as $i'$ decreases, the same strict inequality holds for every smaller remaining $i'$. Therefore, no $k^*$ exists in $S$ for any remaining $i'$, meaning $L^*(i', j) = 2$. Column $j$ is thus correctly filled, and the common differences stored in $F$ during this iteration are precisely those corresponding to arithmetic triples with a middle index, $j$, and the invariant extends to $j$.
 
The outer for loop is finite since each branch of the main while loop strictly decrements $i$ or increments $k$ within bounded ranges, while the second inner while loop strictly decrements $i$, so the algorithm terminates.

By induction, $L[i'][j'] = L^*(i', j')$ for all $i' < j' \leq n - 2$. Every progression of length, $> 2$, in $S$ contains an arithmetic triple having the same common difference, so $F$ contains the collection of common differences of all such progressions, and $d_{max}$ is their maximum.\qed
\end{pf}

%%%%%%%%%%%%%%%%
\begin{theorem}[Time complexity]\label{theorem:complexity-dynamic-programming}
Let $S$ be a strictly increasing sequence of $n$ positive integers and $N$ be the value of its largest integer where $n\leq N$. Then, the time complexity of Alg.~\ref{alg:dynamic-programming-algorithm} is $\bigO(n^2)$.
\end{theorem}

\begin{pf}
The total work by the algorithm is carried out in the procedures detailed below and presented in scope order from top to bottom.
\begin{enumerate}
    \item The outer for loop in lines 5--30 traverses exactly $n-2$ elements of $S$, which is $\bigO(n)$.
    \item For each fixed $j$, the number of iterations of the while loop in lines 8--25 is at most $\bigO(j+(n-j))=\bigO(n)$ since $i$ starts at $j-1$ and decreases to $0$ (when considering also the second while loop in lines 26--29), which is at most $j$ times, at the same time $k$ starts at $j+1$ and increases to $n$, which is at most $n-j-1$ times.
\end{enumerate}
\sloppy From the above, the total time of Alg.~\ref{alg:dynamic-programming-algorithm} can be expressed as 
\[
    T_3\left(n\right)=\bigO(n^2).
\hfill\qed\]
\end{pf}

%%%%%%%%%%%
\subsection{Algorithm 4}\label{sec:factors-alg}

Alg.~\ref{alg:factors} is called in line 4 of Alg.~\ref{alg:proposed-algorithm} and is based on Euler's Linear Sieve algorithm \citep{Gries1978,Ireland1990} for finding the smallest prime factor of each number up to a number, $N$, with a known modification to the standard procedure that allows for more efficiently retrieving from the constructed sieve all distinct prime factors for any given number \cite{Pritchard1987}. This modification takes the form of an array, $n_p$, which contains for any integer, $t = \prod_{i=1}^k p_i^{e_i}$, where $t \ge 2$ and $p_1 < \cdots < p_k$, the $p_1$-free part of $t$, $\frac{t}{p_1^{e_1}}$, obtained by removing the full power of its smallest prime factor. Thus, for all $t$ from $2$ to $d_{max}$, $n_p[t]$ will store this $p_1$-free part.

\begin{algorithm}[!t]
\small
\DontPrintSemicolon
\SetProcNameSty{textsc}
\SetProcArgSty{textsc}
\renewcommand{\SetProgSty}[1]{\renewcommand{\ProgSty}[1]{\textnormal{\csname#1\endcsname{##1}}\unskip}}%
\SetProgSty{textsc}
\SetKwData{Left}{left}\SetKwData{This}{this}\SetKwData{Up}{up}
\SetKwInOut{Input}{input}\SetKwInOut{Output}{output}\SetKwInOut{Initialize}{Initialize:}
\SetKw{Continue}{continue}
\Input{An unordered collection, $F$, of key-value pairs containing an integer key, $d$, and an empty ordered collection for its corresponding value as well as a integer, $d_{max}$, specifying the maximum $d\in F$}
\Output{An unordered collection, $F$, of key-value pairs containing for each key the collection of co-factors of the form $p\mid d$ for its respective integer key}
\SetKwProg{Fn}{Function}{}{end}
\Fn{\factorsprime($F, d_{max}$)}{
$p_s \leftarrow \langle\rangle$\Comment*[r]{Prime numbers up to $d_{max}$}
$s_p \leftarrow \langle\rangle \, \oplus 0, i\leftarrow 0\ldots d_{max}$\Comment*[r]{Smallest prime factor of each $i$}
\nlnonumber
\nlast $n_p \leftarrow \langle\rangle \, \oplus 0, i\leftarrow 0\ldots d_{max}$\Comment*[r]{$p$-free part of i with respect to its smallest prime factor}
\For(\Comment*[f]{Modified Euler Linear Sieve procedure}){$i \leftarrow 2$ \KwTo $d_{max}$}{
    \If(\Comment*[f]{Prime number found}){$s_p[i] == 0$}{
        $p_s\leftarrow p_s \, \oplus i$\Comment*[r]{Store prime number}
        $s_p[i] \leftarrow i$\Comment*[r]{Store smallest prime factor}
        \nlnonumber
        \nlast $n_p[i]\leftarrow 1$
    }
    \ForEach(){$p \in p_s \textnormal{ where } p \leq s_p[i] \textnormal{ and } p \cdot i \leq d_{max}$}{
        $m\leftarrow p \cdot i$\;
        $s_p[m]\leftarrow p$\Comment*[r]{Mark multiple with prime}
        \eIf(\Comment*[f]{Prevent visiting the same composite number twice}){$p == s_p[i]$}{
            \nlnonumber
            \nlast $n_p[m]\leftarrow n_p[i]$\Comment*[r]{Mark multiple with $p$-free part of earlier prime i.e., 1}
        }{
            \nlnonumber
            \nlast $n_p[m]\leftarrow i$\Comment*[r]{Mark multiple with the $p$-free part of $m$}
        }
    }
}
\ForEach(\Comment*[f]{Compute $p\mid d$ co-factors for each common difference discovered in $S$}){$d \in F$}{
    $t \leftarrow d$\Comment*[r]{Pointer in $s_p$ for traversing distinct prime factors}
    \While{$t > 1$}{
        $c \leftarrow \frac{d}{s_p[t]}$\Comment*[r]{$p\mid d$ co-factor}
        \If(\Comment*[f]{$p\mid d$ co-factor must be a discovered common difference in $S$}){$c \in F$}{
            $F[d]\leftarrow F[d] \, \oplus c$\Comment*[r]{Store $p\mid d$ co-factor}
        }
        $t \leftarrow n_p[t]$\Comment*[r]{Update pointer to index of next distinct prime}
    }
}
    \Return $F$
}
\caption{Modified procedure based on Euler's Linear Sieve algorithm \citep{Gries1978} for finding all co-factors, $p\mid d$, where $p$ is a distinct prime factor, for all common differences, $d$, belonging to arithmetic triples within an input sequence. Note that asterisks indicate modifications to the original algorithm as used to enumerate inclusion-maximal arithmetic progressions (IMAPs) with the proposed algorithm shown in Alg.~\ref{alg:proposed-algorithm}.}\label{alg:factors}
\end{algorithm}

%%%%%%%%%%%%%%%%
\begin{theorem}[Correctness]\label{theorem:correctness-factors}
For every finite collection of positive integers, $F$, and each common difference, $d$, in $F$, Alg.~\ref{alg:factors} computes 
exactly the collection, $\lbrace \frac{d}{p} : p \, \textnormal{prime}, p\mid d, \frac{d}{p}\in F\rbrace$.
\end{theorem}

\begin{pf}
Lines 5--20 implement Euler's Linear Sieve, which correctly computes $s_p[t]$ (i.e., the smallest prime factor of $t$) for every $1<t\le d_{\max}$, generating each composite exactly once as $t = p\cdot i$ with $p = s_p[t]$ \citep{Gries1978}. Based on the modification noted above using $n_p$, we establish the following invariant: for every $t>1$ with $s_p[t]=p_1$ and $t=p_1^{e_1}r$ with $p_1\nmid r$, we have $n_p[t]=r$.
 
We prove the invariant by induction on $t$. If $t$ is prime, lines 8--9 set $s_p[t]=t$ and $n_p[t]=1=r$ ($e_1=1$, confirming the invariant). If $t=p_1^{e_1}\cdots p_k^{e_k}$ is composite, it is generated as $t=p_1\cdot i$ with $i=\frac{t}{p_1}<t$. Since $i<t$, the invariant holds for $i$ by induction, and line 13 sets $s_p[t]=p_1$. If $e_1=1$, then $s_p[i]\ne p_1$ (as the smallest prime of $i=p_2^{e_2}\cdots p_k^{e_k}$ is $p_2>p_1$), so line 17 sets $n_p[t]$ equal to $i=\frac{t}{p_1}=r$. However, if $e_1>1$, then $s_p[i]=p_1$, so line 15 sets $n_p[t]$ equal to $n_p[i]=\frac{i}{p_1}^{e_1-1}=\frac{t}{p_1}^{e_1}=r$. For a fixed $d\in F$, the while loop in lines 23--29 traverses $t_0=d,\;t_1=n_p[t_0],\;t_2=n_p[t_1],\ldots$ By the invariant, each step removes the full power of $s_p[t_j]$ from $t_j$, such that consecutive terms correspond to the distinct prime factors of $d$ one at a time, each exactly once, until $t=1$. At each step, $p=s_p[t]$ is a prime factor of $d$ and $c=\frac{d}{p}$ is a valid co-factor by Theorem~\ref{lemma:imap-characterization-via-primes}, so $c$ is stored only if $c\in F$ (lines 25--27). Since every distinct prime factor of $d$ is visited, every value, $\frac{d}{p}\in F$, is reported, and since each prime is visited exactly once, no value is duplicated. 

All loops range over either a finite index collection (i.e., $i=2,3,\ldots, d_{max}$) or finite collections (i.e., $p_s$ and $F$), and each iteration of the while loop in lines 23--29 satisfies $n_p[t]<t$ for all $t>1$, so $t$ strictly decreases to 1, and thus the algorithm terminates.

By the invariant at termination, $F[d]$ contains exactly the collection, $\{\frac{d}{p} : p \textnormal{ prime}, p \mid d, \frac{d}{p} \in F\}$.\qed
\end{pf}

%%%%%%%%%%%%%%%%
\begin{theorem}[Time complexity]\label{theorem:complexity-factors}
Let $F$ be a collection of positive integers corresponding to the distinct common differences, $d$, belonging to arithmetic progression triples discovered in a strictly increasing sequence of $n$ positive integers, $S$, where $N$ is its largest integer and $n\leq N$. Then, the time complexity of Alg.~\ref{alg:factors} is $\bigO \left( N + \min \left( n^2 \frac{ \log N }{ \log \log N } , N \log\log N \right) \right)$.
\end{theorem}

\begin{pf}
The total work by the algorithm is carried out in the procedures detailed below and presented in scope order from top to bottom.

\begin{enumerate}
    \item The Euler Linear Sieve procedure is used in lines 5--20 to compute a standard sieve containing the smallest prime factor of each of the numbers from 2 to $d_{max}\leq N$, which has been shown to be $\bigO(N)$ \citep{Gries1978}. The modification made to this procedure (lines 4, 9, 15, and 17), proven in Theorem~\ref{theorem:correctness-factors}, does not alter the complexity of the original procedure.
    \item The for each loop in lines 21--30 traverses all distinct common differences in $F$ of arithmetic progression triples discovered in $S$ and previously computed in Alg.~\ref{alg:dynamic-programming-algorithm}. Since $S$ has at most $\binom{n}{2}$ pair-wise differences with each a positive integer of at most $N$, the number of such distinct common differences is $\bigO(\min(n^2,N))$. 
    \item For a fixed common difference, $d\in F$, the while loop in lines 23--29 traverses all \textit{distinct} prime factors, $p$ of $d$, from the sieve exactly once (by the proof of Theorem~\ref{theorem:correctness-factors})---the number of which is denoted by $\omega(d)$. Combined with (2), the total cost across all $d$ is $\sum_{d\in F}\omega(d)$. In a way similar to the proof of Theorem~\ref{theorem:complexity-proposed-refined}, we bound this sum in the following two ways.
    \begin{enumerate}
        \item[(i)] \sloppy Every $d\le N$ has $\omega(d)=\bigO\left(\frac{\log N}{\log\log N}\right)$,
         which is the maximal order of $\omega$ \citep{Hardy2008}. Therefore,
        $\sum_{d\in F} \omega(d) =\bigO \left(\min(n^2,N)\frac{\log N}{\log\log N}\right)$.
        \item[(ii)] \sloppy On the other hand, consider the function, $\chi(p) = |\{ d \in F : p \mid d \}|$. Since $F\subseteq\{1,2,\dots,N\}$, we have $\chi( p ) \le \frac{ N }{ p }$. Therefore, we obtain
        \begin{align*}
            \sum_{d\in F} \omega(d) &= \sum_{d \in F} \sum_{p \le N} |\{ (p,d) : p \text{ prime }, p \mid d \}| \\
            &\le \sum_{p \le N} \chi( p ) \\
            &\le \sum_{p\le N} \frac{N}{p} \\
            &= N \sum_{p\le N} \frac{1}{p} \\
            &=N\left(\log\log N + A + o\left(\frac{1}{\log N}\right)\right) \\
            &=N\log\log N + \bigO(N)
        \end{align*}
        by Mertens' second theorem \citep{Hardy1917,Hardy2008,Chen2025}, so
        $\sum_{d\in F} \omega(d) = \bigO(N\log\log N)$.
    \end{enumerate}
    Taking the smaller of (i) and (ii) leaves $\bigO \left(\min\left(n^2\frac{\log N}{\log\log N}, N\log\log N\right)\right)$ since $\min(n^2,N)$ in (i) becomes $n^2$ because when $N<n^2$ the bound, $N\log\log N$, of (ii) becomes the minimum in either case. Finally, for each distinct prime factor, $p$ of $d$, lines 24--27 compute and then store the $p\mid d$ co-factor (proof of Lemma~\ref{lemma:imap-characterization-via-primes}). The number of $p\mid d$ co-factors is equal to the number of prime factors of $d$, so the complexity remains the same.
\end{enumerate}
From the above, the total time of Alg.~\ref{alg:factors} can be expressed as
\[
T_4\left(n, N\right)=\bigO \left( N + \min \left( n^2 \frac{ \log N }{ \log \log N } , N \log\log N \right) \right).\hfill\qed
\]
\end{pf}

%%%%%%%%%%%%%%%%%%%%%%%%%%%%%%%%%%%%%%%%%%%%%%
%%%%%%%%%%%%%% Expected Number %%%%%%%%%%%%%%%
%%%%%%%%%%%%%%%%%%%%%%%%%%%%%%%%%%%%%%%%%%%%%%
\section{Expected Number of Inclusion-Maximal Arithmetic Progressions}\label{sec:expected-number}

In this section, we are concerned with determining the expected number of IMAPs in a random sequence of a given length and sampled from a specified range. Towards this end, we begin by providing definitions needed for first formalizing and then counting the number of so-called \textit{centered}, \textit{low/high}, and \textit{wide progressions}, which we will use to compute the expected number of IMAPs and conclude with bounds on their order of growth. Recall that we will assume here that $S = \langle s_0, s_1, \dots, s_{n-1} \rangle$ with $1 \le s_0 < s_1 < \dots < s_{n-1} \le N$ and drawn uniformly at random from $\binom{N}{n}$ such sequences. While this differed from our earlier assumptions where $s_{n-1} = N$, this change will affect the number of IMAPs by at most only a linear term, which becomes asymptotically negligible.

%%%%%%%%%%%%%%%%%%%%%%
\subsection{Maximal Arithmetic Progressions}\label{sec:maximal-aps}

We will use the notation, $[t_0,d,k]$, for an arithmetic progression containing the elements, $\{t_{0}+di \; | \; 0 \leq i < k \}$. In order to calculate the expected number of IMAPs, we will first determine, for each arithmetic progression, the number of sequences in which they are maximal. There are three different cases below that we need to consider.

%%%%%%%%%%%
\subsubsection{Centered Arithmetic Progressions}

We will refer to any arithmetic progression that can be extended left and right as a \textit{centered arithmetic progression (c-ap)}. For such an arithmetic progression, $t = \langle t_0, t_1, \dots, t_{k-1} \rangle$, with $t_0 > d$ and $N - t_{k-1} \geq d$, we see that it is maximal in a sequence $(s_i)$ if and only if both $t_0 - d$ and $t_{k-1} + d$ are not contained in $(s_i)$. Therefore, such an arithmetic progression is maximal in exactly $\binom{N-k-2}{n-k}$ different sequences $(s_i)$. Let us further fix $0<d<N$ and $2 < k\leq n$. We denote by $T_c(d,k)$ the number of c-ap's of common difference, $d$, and length, $k$.
If $\apt$ is a c-ap then the following expressions hold:
\begin{align*}
	d &< t_0 & \text{ and } && t_{k-1}=t_0+(k-1)d&\leq N-d\\
	&&&& \Leftrightarrow t_0 &\leq N-kd.
\end{align*}
By combining these expressions, we see that $d < t_0 \leq N-kd$. Therefore, $T_c(d,k)=\max \{ 0, N-(k+1)d \}$.

%%%%%%%%%%%%%%%%%%%%%%
\subsubsection{Low and High Arithmetic Progressions}

We will refer to any arithmetic progression that can only be extended on its right or left side as a \textit{low or high arithmetic progression (l-ap or h-ap)}, respectively. Such an arithmetic progression is maximal in exactly $\binom{N-k-1}{n-k}$ sequences. We denote by $T_\ell(d,k)$ the number of l-aps with common difference, $d$, and length, $k$. We see that an arithmetic progression $\apt$ can be only extended to the right if
\begin{align*}
	0< t_0 & \leq d & \text{ and } &&	t_{k-1}=t_0+(k-1)d& \leq N-d\\
	&&&&\Leftrightarrow t_0 &\leq N-kd.
\end{align*}
Combining these expressions, we see that $T_\ell(d,k)=\max\left\{0,\min\{d,N-kd\}\right\}$. By symmetry, the number of h-aps is also given by $T_h(d,k)=\max\left\{0,\min\{d,N-kd\}\right\}$.

%%%%%%%%%%%%%%%%%%%%%%
\subsubsection{Wide Arithmetic Progressions}

We will refer to any arithmetic progression that cannot be extended a \textit{wide arithmetic progression (w-ap)}. Such an arithmetic progression is maximal in exactly $\binom{N-k}{n-k}$ sequences. We denote by $T_w(d,k)$ the number of w-aps with common difference, $d$, and length, $k$. We see that $\apt$ cannot be extended if
\begin{align*}
	0 < t_0 & \leq d & \text{ and } &&	N-d < t_{k-1}=t_0+(k-1)d & \leq N\\
	&&&& \Leftrightarrow N-kd < t_0 & \leq N-(k-1)d.
\end{align*}
By combining these expressions, we see that $T_w(d,k)=\max\{0, d- |N-kd|\}$.

%%%%%%%%%%%%%%%%%%%%%%
\subsection{Inclusion-Maximal Arithmetic Progressions}\label{sec:IMAP-math}

For calculating the expected number of IMAPs in a random sequence, we will count the number of sequences, $(s_i)$, in which a given MAP, $\apt$, is inclusion-maximal. In what follows, we let 
\[N'= 
\begin{cases} N_c=N-2 & \text{ if \apt\ is centered}\\
N_\ell=N_h=N-1 & \text{ if \apt\ is low or high}\\
N_w=N & \text{ if \apt\ is wide.}
\end{cases} \]

\noindent For a fixed arithmetic progression, $\apt$, the number of sequences containing it is given by $|\Sap|= \binom{N-k}{n-k}$, while the number of sequences in which $\apt$ is maximal is given by \[|\Sapm| = \binom{N'-k}{n-k}.\]

We will now work exclusively within the set, $\Sapm$, and determine when the MAP, $\apt$, is inclusion-maximal.

\begin{lemma}[Disjointness of gap-filling sets]\label{lemma:disjointness}
Let $\apt$ be an arithmetic progression and $p$ and $q$ be two different primes dividing $d$. Then, the intersection of the two refinements, $[t_0,\frac{d}{p},p(k-1)+1]$ and $[t_0,\frac{d}{q},q(k-1)+1]$, is just the elements of $\apt$.
\end{lemma}

\begin{pf}
Assume we have an element of the intersection, $x \in [t_0,\frac{d}{p},p(k-1)+1] \cap [t_0,\frac{d}{q},q(k-1)+1]$. Then, $x=t_0+\frac{d}{p}i=t_0+\frac{d}{q}j \Leftrightarrow qi=pj$. It follows that $p\mid i=pm$ and therefore, $x=t_0+dm \in \apt$.\qed
\end{pf}

\begin{rmk}[Prime refinement of inclusion-maximality]
    Given a sequence, $\seq \subseteq \{1,\ldots,N\}$, of length $n$ and a MAP, $\apt \subseteq \seq$, we see that \apt\ is inclusion-maximal if, for all primes, $p\mid d$, we have the following: $T_p(\apt):=\{t_j+\frac{d}{p} \ell \; : \; 0\leq j <k-1, 1\leq \ell <p \} \nsubseteq \seq$. Lemma~\ref{lemma:disjointness} shows that $T_p$ and $T_q$ are disjoint for primes $p \neq q$.
\end{rmk} 

Given an arithmetic progression, $\apt$, we will count the number of sequences in which it is contained but not inclusion-maximal. We define sets of sequences in which $\apt$ is not inclusion-maximal for different reasons and calculate the size of their union using the inclusion-exclusion principle \citep{Roberts2009}. From section~\ref{sec:maximal-aps}, we know the expected number of sequences in which a given arithmetic progression is maximal, and it is from this that we will now deduct the number of sequences in which it is not inclusion-maximal. The arithmetic progression, $\apt$, is not inclusion-maximal in a sequence, \seq, if and only if $T_p \subset \seq$ for at least one prime, $p\mid d$. The set, $T_p$, contains $(k-1)(p-1)$ elements. Therefore, there are $|S_p|=s(\apt,\{p\}) := \binom{N'-k-(k-1)(p-1)}{n-k-(k-1)(p-1)}$ sequences in which $\apt$ is maximal but not inclusion-maximal because the arithmetic progression, $[t_0,\frac{d}{p},p(k-1)+1]$, is also contained in \seq. Next, we define the sets, $\bigcap_i S_{p_i}$, which contain sequences in which $\apt$ is maximal but not inclusion-maximal because all the arithmetic progressions, $[t_0,\frac{d}{p_i},p_i(k-1)+1]$, are also contained in \seq. Since the sets, $T_{p_i}$, are pairwise disjoint, we see that 
\[ \left|\bigcap_i S_{p_i}\right|=s(\apt,\{p_i\}) = \binom{ N' - k - \sum_i |T_{p_i}| } { n - k - \sum_i |T_{p_i}|}.\]

By inclusion-exclusion, the number of sequences in which $\apt$ is maximal but not inclusion-maximal is \[ \left|\bigcup_{\substack{p_i \\ p_i|d}}S_{p_i} \right|= \sum_{J \subseteq \{p_i\}} (-1)^{|J|+1} s(\apt,J). \]

\noindent For a given arithmetic progression, $\apt$, that is centered ($c$), low ($\ell$), high ($h$), or wide ($w$), we define the functions, $M_c(d,k)$, $M_\ell(d,k)$, $M_h(d,k)$, and $M_w(d,k)$, respectively, using each of their corresponding $N'$, that count how many sequences contain $\apt$ as an IMAP as
\[ M_*(d,k)= \binom{N'-k}{n-k} - \left|\bigcup_{\substack{p_i \\ p_i\mid d}}S_{p_i} \right|. \]

%%%%%%%%%%%
\subsubsection{Expected Number and Bounds on Order of Growth}\label{sec:expected-value-bounds}

We can now collect the results from the subsections above and fix the values, $N$ and $m$, for the possible range and size of a sequence, \seq, respectively. We define arithmetic progressions, $\apt$, for any $0<d<\frac{N}{2}$ and $2<k\leq n$ and have defined the number of possible arithmetic progressions for the following parameters:
%\pagebreak
\begin{enumerate}
    \item $T_c(d,k)=\max\{0,N-(k+1)d\}$ (centered);
    \item $T_\ell(d,k)=T_h(d,k)=\max\{0,\min\{d, N-kd\} \}$ (low and high); and
    \item $T_w(d,k)=\max\{0,d-|N-kd|\}$ (wide).
\end{enumerate}

For all sequences with parameters, $N$ and $n$, we calculate the total number of IMAPs that we will encounter when checking all possible sequences as

\[
S^*= \sum_{d=1}^{\lfloor\frac{N-1}{2}\rfloor} \sum_{k=3}^n \big[ T_c(d,k) M_c(d,k) + T_\ell(d,k) M_\ell(d,k) + T_h(d,k) M_h(d,k) + T_w(d,k)M_w(d,k) \big].
\]

\noindent From this, the expected number of IMAPs in random sequences of values from $1$ to $N$ and length $n$ is 
\begin{equation} \nonumber 
    \mathbb{E}[\#\mathrm{IMAP}] = \frac{S^*}{\binom{N}{n}}.
\end{equation}

\noindent Alternatively, the probability that a randomly chosen ordered pair of distinct elements of $S$ begins an IMAP is
\begin{equation} \nonumber 
    P(\mathrm{IMAP}) = \frac{\mathbb{E}[\#\mathrm{IMAP}]}{\binom{n}{2}} = \frac{S^*}{\binom{N}{n}\binom{n}{2}}.
\end{equation}

\begin{theorem}[Order of growth]\label{theorem:expected-number-IMAPs}
The expected number of IMAPs in a sequence, $S$, of $n$ elements in $\{1,2,\ldots,N\}$, where $N = \Theta( n )$, grows as $\Theta( n^2 )$ and the same holds for MAPs.
\end{theorem}

\begin{lemma}[Bounds]\label{lemma:lower-bound-MAPs}
The total number of MAPs in a sequence, $S$, of $n$ elements is at most $\frac{1}{2} \binom{ n }{ 2 }$. Consequently, the number of IMAPs is also at most $\frac{1}{2}\tbinom{n}{2}$.
\end{lemma}

\begin{proofof}{Lemma~\ref{lemma:lower-bound-MAPs}}
Consider the set of ordered pairs, $A = \{ (a,b) : a,b \in S, a < b \}$.  Clearly, for any $(a,b) \in A$, there can only be one MAP in $S$ that begins with $(a,b)$. Therefore, consider the set, $A_0$, of pairs, $(a,b) \in A$, that begin a MAP. Let $A_1 = \{ (a,b) \in A : c = b + (b-a) \in S \}$ and $A_2 = \{ (b,c) \in A : a = b - (c-b) \in S \}$ while $A_0 \subseteq A_1 \setminus A_2$. If $(a, b) \in A_0$, then $(b,c) \in A_2$, where $c$ is the next term of the IMAP, and the mapping, $(a,b) \mapsto (b,c)$, is an injection from $A_0$ to $A_2$. Since $A_0$ and $A_2$ are disjoint, we obtain 
\[
    2 |A_0| \le |A_0| + |A_2| \le |A| = \binom{n}{2}.
\]
Therefore, the number of MAPs is $|A_0|\le\frac{1}{2}\binom{n}{2}$, and since every IMAP is a MAP, the same bound holds for IMAPs.\qed
\end{proofof}

\begin{proofof}{Theorem~\ref{theorem:expected-number-IMAPs}}
The upper bound of $\bigO(n^2)$ is immediate from Lemma~\ref{lemma:lower-bound-MAPs}, so we now establish the matching lower bound. Let $N = \lambda n$ and say, a pair $(a,b = a+d) \in \{1,2,\ldots,N\}^2$ begins a $3$-IMAP if $(a, a+d, a+2d)$ is an IMAP in $S$. If $a + 2d > N$, or equivalently $d > \frac{a + N}{ 2 }$, then $(a,b)$ does not begin a $3$-IMAP. This event occurs with a probability of $\frac{1}{2} + o(1)$, where $o(1)$ is a term that tends to $0$ as $n$ tends to infinity (and is not necessarily positive). The exact probability is given by $P = \lfloor \frac{(N-a)}{2} \rfloor \frac{1}{N-a}$, and with high probability, we can choose $a$ such that $P \ge \frac{1}{2} + o(1)$. Assuming now that $d \le \frac{a + N}{ 2 }$, then $(a,b)$ begins a $3$-MAP if and only if $a + 2d \in S$ and $a - d, a + 3d \notin S$. This event occurs with a probability of $\frac{1}{\lambda^3} + o(1)$. If $(a,b)$ begins a $3$-MAP but not a $3$-IMAP, then there exists a prime number, $p$, such that $(a, a+\frac{d}{p}, \dots, a + (2p-1) \frac{ d }{ p }, a + 2d) \in S$. This event occurs with a probability of $\frac{1}{\lambda^{2(p-1)}} + o(1)$. Therefore, the probability that $(a,b)$ does not begin a $3$-IMAP, ignoring the $o(1)$ term, is at most
\[
    \sum_{p = 2}^\infty \frac{1}{(\lambda^2)^{p-1}} = \frac{1}{\lambda^2 - 1}.
\]
Overall, the probability is at least
\[
    c = \frac{1}{2} \cdot \frac{1}{ \lambda^3 } \cdot \frac{ \lambda^2 - 2 }{\lambda^2 - 1} + o(1).
\]
The expected number of IMAPs of length $3$ is then at least
\[
    c \binom{ N }{ 2 } + o(1) = \frac{ \lambda^2 c }{ 2 }  n^2 + o(1).
\hfill\qed\]
\end{proofof}

As a corollary, we obtain that the expected number of IMAPs that intersect in a given number could be linear in $n$.

\begin{corollary}[Point of maximum density]\label{corrollary:expected-intersecting-IMAPs}
Under Theorem~\ref{theorem:expected-number-IMAPs}, there exists an integer, $a \in \{1,2,\ldots,N\}$, with $N = \Theta( n )$, such that the expected number of IMAPs that intersect in $a$ is $\Theta( n )$. 
\end{corollary}

\begin{rmk}\label{rmk:output-size}
By Lemma~\ref{lemma:lower-bound-MAPs}, the number of IMAPs in any sequence of length, $n$, is at most $\frac{1}{2}\binom{n}{2}$, and by Theorem~\ref{theorem:expected-number-IMAPs}, this $\Theta(n^2)$ bound is met in expectation whenever $N=\Theta(n)$. Therefore, the worst-case output size is $\Theta(n^2)$, and any algorithm that explicitly enumerates all IMAPs must spend at least $\Omega(n^2)$ time merely to write its output. Combined with Proposition~\ref{prop:lower-bound}, Alg.~\ref{alg:proposed-algorithm} is thus both time-optimal and output-size optimal up to a factor of $\log\log n$ when $N=\Theta(n)$.
\end{rmk}

%%%%%%%%%%%%%%%%%%%%%%%%%%%%%%%%%%%%%%%%%%%%%%
%%%%%%%%%%%%%% Experiments %%%%%%%%%%%%%%%%%%%
%%%%%%%%%%%%%%%%%%%%%%%%%%%%%%%%%%%%%%%%%%%%%%
\section{Experiments}\label{sec:experiments}

In this section, we detail the results of a comparative analysis of the practical running times of the proposed algorithm, Alg.~\ref{alg:proposed-algorithm}, and the existing algorithm, Alg.~\ref{alg:existing-algorithm}. In our performance comparisons, we consider integer sequences of two types: (1) varying length, $n$, with a linearly dependent sampled range, $\lambda n$, for some constant, $\lambda\geq 1$, and (2) fixed length, $n$, with a varying sampled range, $\lambda n$, for different $\lambda\geq 1$. We report on the number of IMAPs and running times as a function of both the length and sampled range of these sequences. We conclude by reporting on the number of IMAPs enumerated by Alg.~\ref{alg:proposed-algorithm} and their expected number as calculated in section~\ref{sec:IMAP-math} for integer sequences of the first type as empirical support for our proof and its exact construction.

%%%%%%%%%%%%%%%%%%%%%%
\subsection{Setup and Procedure}

In Experiment 1 comparing the practical running times of Algs.~\ref{alg:existing-algorithm} and \ref{alg:proposed-algorithm}, we used strictly increasing integer sequences of \textit{varying} length, $n$, drawn uniformly at random without replacement from $(0\ldotp\ldotp \lambda n]$, $\forall n \in \lbrace 1,1000,\ldots,20000\rbrace$, and where $\lambda=2.5$. The reported number of IMAPs and times here have been averaged over three runs consisting of three different sampled sequences for each $n$. In our second comparison of these same algorithms in Experiment 2, we instead used strictly increasing integer sequences of \textit{fixed} length, $n=10000$, drawn uniformly at random without replacement from $(0\ldotp\ldotp \lambda n]$, $\forall\lambda\in \lbrace 1.0,1.2,\ldots,10.0\rbrace$. The reported number of IMAPs and times here have been averaged over three runs consisting of three different sampled sequences for each $n$ and $\lambda$ pair. In Experiment 3, we used strictly increasing integer sequences of a smaller length, $n$, drawn uniformly at random without replacement from $(0..\,\lambda n]$, $\forall n \in \lbrace 1,2,\ldots,100\rbrace$, and where $\lambda=2.0$, set by the costly exact evaluation of the expected number expression in section~\ref{sec:expected-value-bounds} rather than by the enumerated output as done in the previous experiments. The reported number of enumerated IMAPs here are those found for $10$ different sampled sequences for each $n$. For every individual sequence within the sets of sequences used across all three experiments, the respective number of IMAPs enumerated by Alg.~\ref{alg:existing-algorithm} were checked against the number enumerated by Alg.~\ref{alg:proposed-algorithm} and matched. Finally, for all experiments we used our own C++ implementation of Alg.~\ref{alg:proposed-algorithm} and the C++ implementation of Alg.~\ref{alg:existing-algorithm} provided in \cite{Kranenburg2013}. Each algorithm was run independently on the same 2023 MacBook Pro 14-inch laptop (macOS Sonoma 14.4; M2 Max arm64 chip; 32 GB RAM) running a single instance of XCode Developer Tools (version 15.2) compiled using GNU++20 with both flag \texttt{-O0} (none) and flag \texttt{-O3} (fastest) optimizations separately enabled. All reported running times are expressed as averaged CPU clock times in seconds using \texttt{<time.h>} and \texttt{clock()} (dividing by \texttt{CLOCKS\_PER\_SEC}) calculated for only the search process, excluding initialization of main variables, of the respective algorithms (i.e., lines 4--15 of Alg.~\ref{alg:existing-algorithm} and lines 5--25 of Alg.~\ref{alg:proposed-algorithm}).

%%%%%%%%%%%%%%%%%%%%%%
\subsection{Results and Discussion}

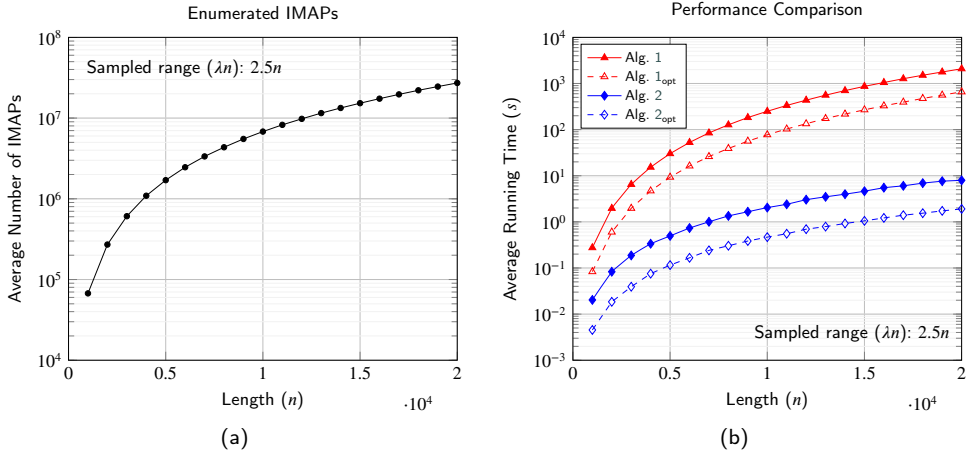
\begin{figure}[]
\centering
\subfloat[]{
\scalebox{0.75}{
 \begin{tikzpicture}
  \begin{axis}[ymode=log, xmin=0,xmax=20000,xlabel={Length ($n$)},ymin=10000,ymax=100000000,ylabel={Average Number of IMAPs},title={Enumerated IMAPs},grid=both,
        grid style={line width=.1pt, draw=gray!10},
        major grid style={line width=.2pt,draw=gray!50}]
        \addplot[mark=*,mark size=1.35pt, black, tension={0.15}] table {figs/num_of_imaps.txt};

        \node[below right, align=center, text=black]
        at (rel axis cs:0.03,0.95) {Sampled range ($\lambda n$): $2.5n$};
  \end{axis}
\end{tikzpicture}
}
}%
\subfloat[]{
\scalebox{0.75}{
\begin{tikzpicture}
  \begin{axis}[ymode=log,
        xmin=0,xmax=20000,xlabel={Length ($n$)},ymin=0.001,ymax=10000,ylabel={Average Running Time ($s$)},title={Performance Comparison},legend style={at={(0.02,0.98)}, anchor=north west}, legend style={nodes={scale=0.8, transform shape}},
        legend cell align={left},
        grid=both,
        mark options={solid},
        grid style={line width=.1pt, draw=gray!10},
        major grid style={line width=.2pt,draw=gray!50}]

        \addplot[mark=triangle*, mark size=2pt, smooth, red, tension={0.15}] table {figs/ima_times_20000_no_opt.txt};
        \addplot[mark=triangle, mark size=2pt, dashed, red, tension={0.15}] table {figs/ima_times.txt};
         
        \addplot[mark=diamond*, mark size=2.2pt, smooth, blue, tension={0.15}] table {figs/imap_times_20000_no_opt.txt};
        \addplot[mark=diamond, mark size=2.2pt, dashed, blue, tension={0.15}] table {figs/imap_times.txt};

        \node[below right, align=center, text=black]
        at (rel axis cs:0.45,0.13) {Sampled range ($\lambda n$): $2.5n$};
        \legend{Alg.~\ref{alg:existing-algorithm},Alg.~\ref{alg:existing-algorithm}\textsubscript{opt},Alg.~\ref{alg:proposed-algorithm},Alg.~\ref{alg:proposed-algorithm}\textsubscript{opt}}
  \end{axis}
\end{tikzpicture}
}
}%
\caption[]{\textbf{Experiment 1:} Average number of inclusion-maximal arithmetic progressions (IMAPs) in (a) and average running times in seconds to enumerate them using Alg.~\ref{alg:existing-algorithm} (red triangles, solid line) and Alg.~\ref{alg:proposed-algorithm} (blue diamonds, solid line) in (b) for strictly increasing integer sequences of varying length, $n$, drawn uniformly at random without replacement from $(0\ldotp\ldotp \lambda n]$, $\forall n \in \lbrace 1000,2000,\ldots,20000\rbrace$ and where $\lambda=2.5$. In (b), Alg.~\ref{alg:existing-algorithm}\textsubscript{opt} (red triangles, dashed line) and Alg.~\ref{alg:proposed-algorithm}\textsubscript{opt} (blue diamonds, dashed line) are the respective average running times for the same sequences using compiler flag \texttt{-O3} (faster) optimized builds. Note that the reported number of IMAPs and times are averaged over three different sampled sequences for each $n$.}\label{fig:performance-comparison-varying-n}
\end{figure}

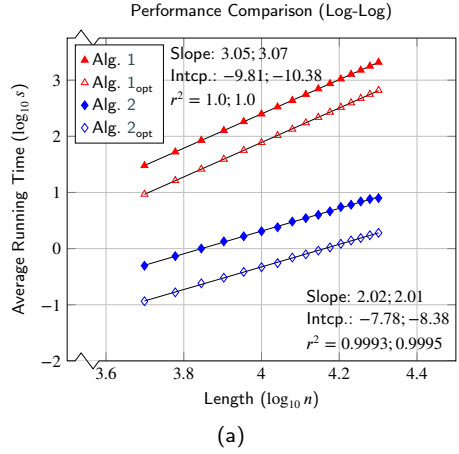
\begin{figure}[]
\centering
\subfloat[]{
\scalebox{0.75}{
\begin{tikzpicture}
  \begin{axis}[
        xmin=3.5,xmax=4.5,xlabel={Length ($\log_{10}n$)},ymin=-2,ymax=3.75,ylabel={ Average Running Time ($\log_{10}s$)},title={Performance Comparison (Log-Log)},legend style={at={(0.02,0.98)}, anchor=north west}, style={nodes={scale=0.95, transform shape}}, legend cell align={left}, axis x discontinuity=crunch, grid=both,
        grid style={line width=.1pt, draw=gray!10},
        major grid style={line width=.2pt,draw=gray!50}]

        % Alg 1 and 2 points
        \addplot[only marks, mark=triangle*, mark size=2pt, red] table {figs/ima_log_no_opt.txt};
        \addplot[only marks, mark=triangle, mark size=2pt, red] table {figs/ima_log.txt};

        \addplot[only marks, mark=diamond*, mark size=2.2pt, blue] table {figs/imap_log_no_opt.txt};
        \addplot[only marks, mark=diamond, mark size=2.2pt, blue] table {figs/imap_log.txt};

        % Alg 1 and 2 lines
        \addplot[smooth, black] table {figs/ima_line_no_opt.txt};
        \addplot[smooth, black] table {figs/ima_line.txt}; 
     
        \addplot[smooth, black] table {figs/imap_line_no_opt.txt}; 
        \addplot[smooth, black] table {figs/imap_line.txt}; 
      
        \legend{Alg.~\ref{alg:existing-algorithm},Alg.~\ref{alg:existing-algorithm}\textsubscript{opt}, Alg.~\ref{alg:proposed-algorithm},Alg.~\ref{alg:proposed-algorithm}\textsubscript{opt}}

        \node[right, align=center, text=black, font=\small]
        at (rel axis cs:0.25,0.95) {Slope: $3.05; 3.07$};
        \node[right, align=center, text=black, font=\small]
        at (rel axis cs:0.25,0.88) {Intcp.: $-9.81; -10.38$};
        \node[right, align=center, text=black, font=\small]
        at (rel axis cs:0.25,0.81) {$r^2=1.0; 1.0$};
        
        \node[below right, align=center, text=black, font=\small]
        at (rel axis cs:0.6,0.24) {Slope: $2.02; 2.01$};
        \node[below right, align=center, text=black, font=\small]
        at (rel axis cs:0.6,0.17) {Intcp.: $-7.78; -8.38$};  
        \node[below right, align=center, text=black, font=\small]
        at (rel axis cs:0.6,0.1) {$r^2=0.9993; 0.9995$};

  \end{axis}
\end{tikzpicture}
}
}%
\caption[]{Log-log plot with linear regression fitted lines of the average running times taken from Experiment 1 (shown in Figure~\ref{fig:performance-comparison-varying-n}) for sufficiently large $n$ of both un-optimized (top) and optimized (bottom) versions of Alg.~\ref{alg:existing-algorithm} (red triangles) and Alg.~\ref{alg:proposed-algorithm} (blue diamonds). Note that compiler flag \texttt{-O3} (faster) optimized builds have been used for the optimized versions of each algorithm.}\label{fig:performance-comparison-varying-n-log}
\end{figure}

Figure~\ref{fig:performance-comparison-varying-n} shows the results from Experiment 1 of the actual number of IMAPs in (a) and averaged running times of Algs.~\ref{alg:existing-algorithm} and \ref{alg:proposed-algorithm} in (b) for integer sequences of the type, $(0\ldotp\ldotp 2.5n]$, $\forall n \in \lbrace 1,1000,\ldots,20000\rbrace$. The average number of actual IMAPs exhibits quadratic growth, in line with the proof provided in Lemma~\ref{lemma:lower-bound-MAPs}, as the length of the integer sequence increases linearly and the sampled range is held to a constant factor, $\lambda=2.5$, of $n$---approaching $2.8\cdot 10^7$ IMAPs for sequences of length $20000$. The empirical quadratic growth rate in the number of IMAPs is in line with similar findings for the number of arithmetic triples, $\bigO(\frac{n^2}{2})$ \citep{Green2008}, and MAPs, $\bigO(\frac{n^2}{k^2})$ \cite{Erickson1999}. In line with expectations regarding their respective theoretical complexities, Alg.~\ref{alg:existing-algorithm} clearly exhibits a considerably higher order of growth in its running time when compared to Alg.~\ref{alg:proposed-algorithm}. Notably, Alg.~\ref{alg:proposed-algorithm} is approximately $260$ times faster than Alg.~\ref{alg:existing-algorithm} for sequences of the longest tested length (i.e., $n=20000$)---topping out at approximately $8$ versus $2083$ seconds. In (c), it is evident that rather aggressive compiler optimizations significantly lower the average runs times of both algorithms, but this effect is considerably greater for Alg.~\ref{alg:proposed-algorithm}\textsubscript{opt}, which tops out at approximately $1.9$ seconds for sequences of the longest tested length versus $658$ seconds for Alg.~\ref{alg:existing-algorithm}\textsubscript{opt}. These results indicate that the proposed algorithm is approximately $346$ times faster than the existing algorithm for the longest tested sequences.

As a more formal means for empirically supporting the respective theoretical complexities of both algorithms, Figure~\ref{fig:performance-comparison-varying-n-log} shows log-log plots with linear regression fitted lines of their running times (from Experiment 1 shown in Figure~\ref{fig:performance-comparison-varying-n}) for sufficiently large $n$---in this case, $\forall n \in \lbrace 5000,6000,\ldots,20000\rbrace$. As indicated by the slopes of their respective fitted lines in Figure~\ref{fig:performance-comparison-varying-n-log}(a), the empirical complexity of Alg.~\ref{alg:existing-algorithm} is $\bigO(n^{3.05})$ and $\bigO(n^{2.02})$ for Alg.~\ref{alg:proposed-algorithm} with all variability perfectly or near-perfectly explained for both ($r^2=1.0000$ and $r^2=0.9993$, respectively). The fitted lines further cross as indicated by their $y$-intercepts of $-9.81$ for Alg.~\ref{alg:existing-algorithm} and $-7.78$ for Alg.~\ref{alg:proposed-algorithm}. Notably, the empirical complexities for both algorithms are approaching their respective nearly-quadratic and nearly-cubic theoretical counter parts and what would be expected for such sequences in which $n$ and $N$ are moderately close (in this case, when $\lambda=2.5$ is a small constant factor greater than $n$). Finally, we can see that while the slope and thus empirical complexity of Alg.~\ref{alg:existing-algorithm}\textsubscript{opt} remains relatively the same as its un-optimized version, the empirical complexity of Alg.~\ref{alg:proposed-algorithm}\textsubscript{opt} is even closer to its un-optimized counterpart, scaling at $n^{2.01}$, despite the lower $y$-intercept of $-8.38$ which aligns with its overall faster running time.

%%%%%%%%%%%%

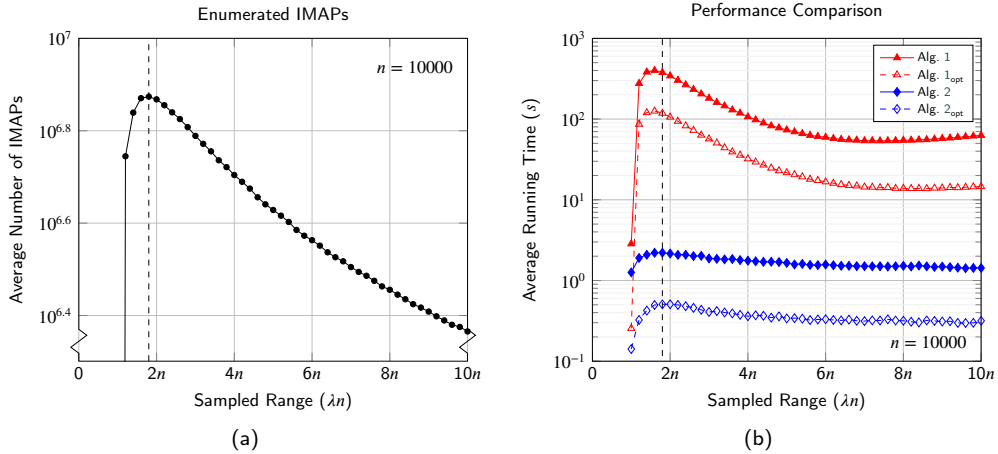
\begin{figure}[]
\centering
\subfloat[]{
\scalebox{0.75}{
 \begin{tikzpicture}
  \begin{axis}[ymode=log,
        xmin=0,xmax=10, xticklabels={0, 0, 2$n$, 4$n$, 6$n$, 8$n$, 10$n$}, xlabel={Sampled Range ($\lambda n$)},
        ymin=2e6,ymax=1e7,
        ylabel={Average Number of IMAPs}, title={Enumerated IMAPs}, grid=both,
        grid style={line width=.1pt, draw=gray!10},
        axis y discontinuity=crunch,
        major grid style={line width=.2pt,draw=gray!50}]
        
        \addplot[mark=*,mark size=1.35pt, black, tension={0.15}] table {figs/num_of_imaps_fixed.txt};
      
        \draw[dashed] (rel axis cs:0.18,0) -- (rel axis cs:0.18,1);
      
        \node[below right, align=center, text=black]
        at (rel axis cs:0.75,0.95) {$n=10000$};
        
  \end{axis}
\end{tikzpicture}
}
}%
\subfloat[]{
\scalebox{0.75}{
\begin{tikzpicture}
  \begin{axis}[ymode=log,
        xmin=0,xmax=10, xticklabels={0, 0, 2$n$, 4$n$, 6$n$, 8$n$, 10$n$}, xlabel={Sampled Range ($\lambda n$)},ymin=0.1,ymax=1000,ylabel={Average Running Time ($s$)}, title={Performance Comparison},legend style={at={(0.98,0.98)}, legend style={nodes={scale=0.72, transform shape}}, legend cell align={left}, anchor=north east}, grid=both,
        mark options={solid},
        grid style={line width=.1pt, draw=gray!10},
        major grid style={line width=.2pt,draw=gray!50}]
 
        \addplot[mark=triangle*, mark size=2pt, solid, red, tension={0.15}] table {figs/ima_times_10000_no_opt.txt};
        \addplot[mark=triangle, mark size=2pt, dashed, smooth, red, tension={0.15}] table {figs/ima_times_fixed.txt};

        \addplot[mark=diamond*, mark size=2.2pt, solid, blue, tension={0.15}] table {figs/imap_times_10000_no_opt.txt};
        \addplot[mark=diamond, mark size=2.2pt, dashed, smooth, blue, tension={0.15}] table {figs/imap_times_fixed.txt};

        \draw [dashed] (18, -10) -- (18, 410);
      
        \legend{Alg.~\ref{alg:existing-algorithm},Alg.~\ref{alg:existing-algorithm}\textsubscript{opt},Alg.~\ref{alg:proposed-algorithm},Alg.~\ref{alg:proposed-algorithm}\textsubscript{opt}}

        \node[below right, align=center, text=black]
        at (rel axis cs:0.75,0.1) {$n=10000$};
  \end{axis}
\end{tikzpicture}
}
}%
\caption[]{\textbf{Experiment 2:} Average number of inclusion-maximal arithmetic progressions (IMAPs) in (a) and average running times in seconds to enumerate them using Alg.~\ref{alg:existing-algorithm} (red triangles, solid line) and Alg.~\ref{alg:proposed-algorithm} (blue diamonds, solid line) in (b) for strictly increasing integer sequences of fixed length, $n=10,000$, drawn uniformly at random without replacement from $(0\ldotp\ldotp \lambda n]$, $\forall \lambda \in \lbrace 1.0,1.2,\ldots,10.0\rbrace$. In (b), Alg.~\ref{alg:existing-algorithm}\textsubscript{opt} (red triangles, dashed line) and Alg.~\ref{alg:proposed-algorithm}\textsubscript{opt} (blue diamonds, dashed line) are the respective average running times for the same sequences using compiler flag \texttt{-O3} (faster) optimized builds. Note that the reported number of IMAPs and times are averaged over three different sampled sequences for each $n$ and $\lambda$ pair.}
\label{fig:performance-comparison-fixed-n}
\end{figure}

As the number of IMAPs in a sequence depends not only on the length of the sequence but on the sampled range from which its integers are taken, we investigated the effect that varying the sampled range for fixed-length sequences has on both the resulting number of IMAPs and the respective running times of each algorithm. Figure~\ref{fig:performance-comparison-fixed-n} shows the results from Experiment 2 of the actual number of IMAPs in (a) and averaged running times of Algs.~\ref{alg:existing-algorithm} and \ref{alg:proposed-algorithm} in (b) for integer sequences of the type, $(0\ldotp\ldotp \lambda n]$, $\forall \lambda \in \lbrace 1.0,1.2,\ldots,10.0\rbrace$ and $n=10000$. One will note in (a) that for $\lambda=1.0$, the sampled sequence is simply all integers from $1$ to $n$, forming a single IMAP equal to the sequence itself. In contrast, for arithmetic progression triples, such dense sets forming an interval (or close to one) have been shown to produce the maximum number \citep{Green2008}. As $\lambda$ increases, the number of IMAPs steeply increases until it reaches its maximum of roughly $7.5\cdot 10^6$ IMAPs for sequences sampled from a range of approximately $\lambda=1.8$ times the sequence length as indicated by the dashed vertical line. In (b), the running times for both algorithms and their optimized versions increase as the number of IMAPs increases despite the fixed sequence lengths, suggesting an output sensitive nature to their constructions, which is common of enumeration algorithms. In line with Figure~\ref{fig:performance-comparison-varying-n}, the running time of Alg.~\ref{alg:existing-algorithm} remains consistently above that of Alg.~\ref{alg:proposed-algorithm} in Figure~\ref{fig:performance-comparison-fixed-n}(b), peaking at approximately $400$ seconds and $126$ seconds (optimized) for sequences in which the sampled range is roughly $\lambda=1.6$ times the sequence length. This is approximately $181$ times and $248$ times higher than the maximum running times, $2.2$ seconds and $0.5$ (optimized) seconds, of Alg.~\ref{alg:proposed-algorithm}, respectively, which occurs at roughly $\lambda=1.8$ and $\lambda=2$ times the sequence length. Notably, it is generally not the case that the highest running time of either algorithm corresponds to the sampled range producing the greatest number of IMAPs (i.e., approximately $\lambda=1.8$ times the sequence length), indicating that there is some trade-off between the number of IMAPs versus their respective maximal lengths and the impact on the overall running time. While the peak running time of Alg.~\ref{alg:existing-algorithm} is significantly higher, its steepness drops off considerably quicker than Alg.~\ref{alg:proposed-algorithm} to a point closer to, but still considerable higher than, the proposed algorithm. This observation, combined with the manyfold increase in peak running time, suggests that Alg.~\ref{alg:existing-algorithm} is more highly affected by the number of IMAPs. Moreover, its running time also begins to slightly increase again as the sampled range becomes greater than approximately $7n$, despite a decreasing number of IMAPs, while the same observation is not quite evident from Alg.~\ref{alg:proposed-algorithm}---suggesting that Alg.~\ref{alg:existing-algorithm} is more highly sensitive to input characteristics of a sequence beyond length, namely, the sampled range from which its integers are drawn or more precisely, the largest integer, $N$, in a given sequence relative to its smallest.

%%%%%%%%%%%%

Finally, Figure~\ref{fig:expected-number-results} shows the results from Experiment 3 of the actual number of IMAPs and the expected number of IMAPs given by our proof and construction (section~\ref{sec:IMAP-math}) for integer sequences of the type, $(0..\,2n]$, $\forall n \in \lbrace 1,2,\ldots,100\rbrace$. Note that the actual number of IMAPs have been enumerated by Alg.~\ref{alg:proposed-algorithm} for $10$ different sampled sequences for each $n$. Not surprisingly, there is a greater variability in the number of possible IMAPs as the length of the sequence increases, centering around approximately $730$ IMAPs for sequences of length $100$. Nonetheless, the smooth curve formed by the expected number of IMAPs centers clearly within the spread of all actual enumerated IMAPs for all $n$, providing empirical support for the correctness of our proof and construction. 

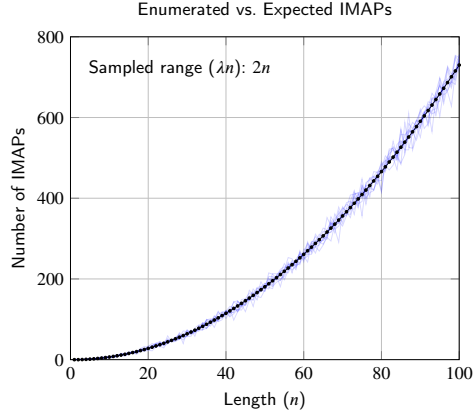
\begin{figure}[]
\centering
\scalebox{0.75}{
 \begin{tikzpicture}
  \begin{axis}[
        xmin=0,xmax=100,xlabel={Length ($n$)},ymin=0,ymax=800,ylabel={Number of IMAPs}, title={Enumerated vs. Expected IMAPs}, grid=both,
        grid style={line width=.1pt, draw=gray!10},
        major grid style={line width=.2pt,draw=gray!50}]
        \addplot[mark=*,mark size=0.65pt, black] table [x=n, y=expected] {figs/expected_imaps.txt};
        \addplot[scatter, no marks, draw=blue, opacity = 0.15] table [x=n, y=run1] {figs/runs_imaps.txt};
        \addplot[scatter, no marks, draw=blue, opacity = 0.15] table [x=n, y=run2] {figs/runs_imaps.txt};
        \addplot[scatter, no marks, draw=blue, opacity = 0.15] table [x=n, y=run3] {figs/runs_imaps.txt};
        \addplot[scatter, no marks, draw=blue, opacity = 0.15] table [x=n, y=run4] {figs/runs_imaps.txt};
        \addplot[scatter, no marks, draw=blue, opacity = 0.15] table [x=n, y=run5] {figs/runs_imaps.txt};
        \addplot[scatter, no marks, draw=blue, opacity = 0.15] table [x=n, y=run6] {figs/runs_imaps.txt};
        \addplot[scatter, no marks, draw=blue, opacity = 0.15] table [x=n, y=run7] {figs/runs_imaps.txt};
        \addplot[scatter, no marks, draw=blue, opacity = 0.15] table [x=n, y=run8] {figs/runs_imaps.txt};
        \addplot[scatter, no marks, draw=blue, opacity = 0.15] table [x=n, y=run9] {figs/runs_imaps.txt};
        \addplot[scatter, no marks, draw=blue, opacity = 0.15] table [x=n, y=run10] {figs/runs_imaps.txt};
  \end{axis}
  \node[below right, align=center, text=black]
        at (rel axis cs:0.03,0.95) {Sampled range ($\lambda n$): $2n$};
\end{tikzpicture}
}
\caption[]{\textbf{Experiment 3:} The number of inclusion-maximal arithmetic progressions (IMAPs) enumerated by Alg.~\ref{alg:proposed-algorithm} (blue lines) and expected number of IMAPs (black circles) as calculated in section~\ref{sec:expected-value-bounds} for strictly increasing integer sequences of varying length, $n$, drawn uniformly at random without replacement from $(0..\,\lambda n]$, $\forall n \in [1,2,\ldots,100]$ and where $\lambda=2.0$. Note that the number of IMAPs enumerated by Alg.~\ref{alg:proposed-algorithm} are those found for $10$ different sampled sequences for each $n$.}
\label{fig:expected-number-results}
\end{figure}

%%%%%%%%%%%%%%%%%%%%%%%%%%%%%%%%%%%%%%%%%%%%%%
\section{Conclusion}

In this paper, we provided an improved $\bigO\left( n^2 \frac{ \log N }{ \log \log N } + N \right)$ enumeration algorithm over an existing algorithm of $\bigO(N^{2+o(1)}n)$ for solving the problem of enumerating inclusion-maximal arithmetic progressions (IMAPs) in a given integer sequence. We provided a mathematical proof for the expected number of IMAPs in random sequences and showed that their order of growth is $\bigO(n^2)$ in Lemma~\ref{lemma:lower-bound-MAPs}. Importantly, this suggests that our proposed nearly-quadratic algorithm is near-optimal when $N=\Theta(n)$ under Corollary~\ref{sec:cor-complexity} and Proposition~\ref{prop:lower-bound}. Our empirical experiments demonstrated an up-to $346$ times improvement in practical running time performance of the proposed algorithm over the existing algorithm for compiler optimized versions and $260$ times improvement for un-optimized versions for sequences of length, $n=20000$, and sampled range, $2.5n$. Moreover, given their asymptotic complexities, we would expect even further improvement as the sequence length continues to grow beyond the tested limit. In further comparisons with sequences of fixed length, the performance of the proposed algorithm was less significantly affected than was the existing algorithm by the sampled range from which integers of a sequence were drawn. Finally, the proposed algorithm will prove useful to mathematical and computational music analysis insofar as being able to analyze large and rhythmically complex music pieces can now be done more efficiently. In future work, it remains to be determined whether more efficient algorithms exist for when either the minimum $k>2$ is fixed or, in the more general case, it can vary with respect to $n$. Moreover, open problems remain with regard to formulating a precise asymptotic constant and tighter bounds, or measures of concentration on the expected number of IMAPs, and exploring how these relate more generally to other k-term arithmetic progressions.


%\section{}\label{}


\bibliographystyle{cas-model2-names}

\bibliography{cas-refs}



\end{document}